\documentclass[a4paper,onecolumn,11pt,accepted=2023-11-03]{quantumarticle}
\pdfoutput=1

\usepackage{hyperref}
\usepackage{url}
\usepackage[left=3cm,top=3cm,right=3cm,bottom=3cm]{geometry}
\usepackage{amsmath}
\usepackage{amsfonts}
\usepackage{amsthm}
\usepackage{graphicx}
\usepackage{caption}
\usepackage{subcaption}
\usepackage{float}
\restylefloat{table}
\usepackage{siunitx}
\usepackage{bbm}
\usepackage{comment}

\newtheorem{theorem}{Theorem}
\newtheorem{problem}[theorem]{Problem}
\newtheorem{definition}[theorem]{Definition}
\newtheorem{lemma}[theorem]{Lemma}
\newtheorem{corollary}[theorem]{Corollary}
\newtheorem{algorithm}[theorem]{Algorithm}

\DeclareMathOperator{\Tr}{Tr}
\DeclareMathOperator{\im}{Im}

\begin{document}

\title{An Improved Approximation Algorithm for Quantum Max-Cut on Triangle-Free Graphs}
\author{Robbie King}
\affiliation{Department of Computing and Mathematical Sciences, Caltech}
\orcid{0000-0002-8152-6340}
\maketitle

\begin{abstract}
    We give an approximation algorithm for Quantum Max-Cut which works by rounding a semi-definite program (SDP) relaxation to an entangled quantum state. The SDP is used to choose the parameters of a variational quantum circuit. The entangled state is then represented as the quantum circuit applied to a product state. It achieves an approximation ratio of $0.582$ on triangle-free graphs. The previous best algorithms of Anshu, Gosset, Morenz \cite{AGM}, and Parekh, Thompson \cite{Parekh-Thompson} achieved approximation ratios of $0.531$ and $0.533$ respectively. In addition we study the EPR Hamiltonian, whose terms project onto EPR states rather than singlet states. (EPR are initials Einstein, Podolsky and Rosen.) We argue this is a natural intermediate problem which isolates some key quantum features of local Hamiltonian problems. For the EPR Hamiltonian, we give an approximation algorithm with approximation ratio $1 / \sqrt{2}$ on all graphs.
\end{abstract}

%\tableofcontents
%\pagebreak

\section{Introduction}

Constraint satisfaction is the canonical NP-complete problem. Since it is impossible to solve these exactly in polynomial time (assuming $\text{P} \neq \text{NP}$), approximation algorithms have been developed. These algorithms often work by relaxing the problem to an efficiently solvable semi-definite program (SDP). They then `round' the SDP solution back to a legal configuration which satisfies as many constraints as possible. The seminal example of such an algorithm is the Goemans-Williamson `hyperplane rounding' algorithm for Max-Cut \cite{GW}, which achieves an approximation ratio of 0.878. The PCP theorem tells us that there is some constant factor of the optimum within which it is NP-hard to find a solution. Stronger computational assumptions such as the Unique Games Conjecture can even be used to show optimality of SDP rounding approaches \cite{GW_UGC_hardness}.

The quantum analogue of constraint satisfaction is the local Hamiltonian problem. Here, the goal is to find the optimum eigenvalue of a given local Hamiltonian. This is the canonical QMA-complete problem, and one can similarly ask for approximation algorithms \cite{QMA_approximation}. In many contexts, `optimum' means the minimum eigenvalue, though when studying approximation algorithms it is more natural to search for the maximum eigenvalue. We will focus on Hamiltonians expressed as a sum of positive semi-definite terms. Then the figure of merit is the energy $E$ of the quantum state output by the algorithm. The approximation ratio is defined as $E/||H||$, where $||H||$ is the spectral norm of the Hamiltonian. A quantum PCP Hamiltonian, conjectured to exist, would have the feature that there is some constant ratio within which it is QMA-hard to approximate $||H||$. Thus searching for better approximation algorithms can be viewed as complementary to searching for quantum PCP Hamiltonians.

The local Hamiltonian problem is of high importance in condensed matter physics and quantum chemistry. In these contexts it corresponds to finding the ground energy of a quantum physical system.

An aspect unique to the quantum scenario is the presence of entanglement in quantum ground states. Not only is the optimal energy QMA-hard to compute, but the optimal state may even be impossible to describe succinctly. Most classical algorithms are limited to search over a small subset of quantum states, given by a choice of ansatz. The natural ansatz is the set of product states, which have no entanglement. In physics this corresponds to mean-field theory, and in quantum chemistry this is related to the Hartree-Fock method. There are many situations in which product states can achieve close to optimal energy \cite{Brandao-Harrow}. However, for a QMA-complete problem, entanglement is necessary. Assuming $\text{NP}\neq\text{QMA}$ there must be instances where product states cannot be inverse polynomially close to optimal, and additionally assuming quantum PCP there must be instances where product states cannot be O(1)-close to optimal.

Quantum Hamiltonian problems also have SDP relaxations. Work has been done on how to round SDP solutions to product states \cite{Brandao-Harrow,BGKT,Gharibian-Parekh,Parekh-Thompson-2-local,Parekh-Thompson-product}. However, we can hope to do better by rounding the SDP to an entangled state in general, as in \cite{AGM,Parekh-Thompson,AGMS}. In this work we try to make progress on this difficult task. Our algorithm can be viewed as a rounding of the SDP to an entangled ansatz.

Our work can also be viewed through the lens of variational quantum algorithms. The ansatz we use is a variational quantum circuit applied to an initial product state. The circuit consists of commuting entangling gates, following \cite{AGM,AGMS}. Using the SDP relaxation, we are able to choose the circuit parameters efficiently and classically. To our knowledge, this is the first time an SDP has been used to optimize a variational quantum circuit.

The high level idea is as follows. Given a Quantum Max-Cut Hamiltonian, it is known how to round an SDP relaxation to a product state which achieves good energy \cite{Gharibian-Parekh}. We would then like to improve the energy with a quantum circuit by entangling the qubits in the support of each local Hamiltonian term. The issue that this runs into is \emph{monogamy of entanglement} \cite{AGM}. That is, we cannot simultaneously entangle the qubits of overlapping terms. Here, we use the SDP again. If the relaxation is sufficiently tight, it will `know about' the constraint of monogamy of entanglement \cite{Parekh-Thompson}. It will then be able to tell us which terms to prioritize when we look to improve the energy of the product state with an entangling circuit.

Monogamy of entanglement is a crucial and elusive extra constraint which is not present in classical constraint satisfaction. The Quantum Max-Cut Hamiltonian is a sum of terms which project pairs of qubits onto the singlet state $\frac{1}{\sqrt{2}} (|01\rangle - |10\rangle)$. It seems to possess the combinatorial difficulty of classical Max-Cut, in addition to the entanglement difficulty. Motivated by this, we define a new Hamiltonian problem which instead consists of terms which project onto the EPR state $\frac{1}{\sqrt{2}} (|00\rangle + |11\rangle)$. Optimisation of this Hamiltonian over product states is trivial, since the optimal product state is simply $|0 \dots 0\rangle$. Thus all the difficulty in this problem arises from `distributing' entanglement in the optimal way.

\section{Results}

We study a particular family of quantum Hamiltonians known as Quantum Max-Cut (QMC). This is defined on a weighted graph $(V, \{w_{ij}\})$ with non-negative weights $w_{ij}\geq0$. Each vertex is a qubit, and to each edge is associated the Hamiltonian term $h = 2 \psi^- = 2 |\psi^-\rangle\langle\psi^-|$ where $|\psi^-\rangle = \frac{1}{\sqrt{2}} \left(|01\rangle - |10\rangle\right)$ is the singlet state. In the Pauli basis $h = \frac{1}{2} \left(\mathbbm{1} - X \otimes X - Y \otimes Y - Z \otimes Z \right)$. The total Hamiltonian is the weighted sum $H = \sum_{ij} w_{ij} h_{ij}$. It has the nice property that it is rotation-invariant; that is, for any single-qubit unitary $U$, we have $(U^\dag \otimes U^\dag) h (U \otimes U) = h$. We have $H \succeq 0$ since the weights are all non-negative. (The symbol $\succeq$ denotes that a matrix is positive semi-definite.) Let's normalise so that we have total weight $\sum_{ij} w_{ij} = 1$.

The maximum-energy state of the QMC Hamiltonian is the same as the groundstate of the anti-ferromagnetic Heisenberg model. This is a Hamiltonian which has been well-studied in physics for decades, serving as a model for magnetism. However, anti-ferromagnetic Heisenberg is usually written as a traceless Hamiltonian, whilst QMC is shifted so that it is positive semi-definite. This means approximation ratios for the two may differ.

\begin{problem}\label{problem_QMC}
\emph{(QMC)}
Given a weighted graph $(V, \{w_{ij}\})$, compute the maximum eigenvalue $||H||$ of the Hamiltonian defined above.
\end{problem}

We remark that if we had instead used the diagonal Hamiltonian term $h = \frac{1}{2} \left(\mathbbm{1} - Z \otimes Z \right)$, this would correspond to the original classical Max-Cut problem.

We also study another family of Hamiltonians, which we will call the EPR Hamiltonian (EPR). Also defined on a weighted graph $(V, \{w_{ij}\})$ with non-negative weights $w_{ij}\geq0$, to each edge is associated the Hamiltonian term $g = 2 \phi^+ = 2 |\phi^+\rangle\langle\phi^+|$ where $|\phi^+\rangle = \frac{1}{\sqrt{2}} \left(|00\rangle + |11\rangle\right)$ is the EPR pair. In the Pauli basis $g = \frac{1}{2} \left(\mathbbm{1} + X \otimes X - Y \otimes Y + Z \otimes Z \right)$. The total Hamiltonian is the weighted sum $G = \sum_{ij} w_{ij} g_{ij}$. $G \succeq 0$ since the weights are all non-negative.

\begin{problem}\label{problem_EPR}
\emph{(EPR)}
Given a weighted graph $(V, \{w_{ij}\})$, compute the maximum eigenvalue $||G||$ of the Hamiltonian defined above.
\end{problem}

QMC is known to be QMA-hard \cite{QMA_complete,QMA_complete_2}. However, there is no known hardness result if the graph is forced to be triangle-free.

There is no hardness result for EPR. Switching momentarily back to energy minimization, the Hamiltonian $-G$ is \emph{stoquastic}, in other words sign-problem-free. Thus it belongs to the complexity class StoqMA, which is believed to be strictly smaller than QMA \cite{QMA_complete,QMA_complete_2}. It is possible that EPR is indeed solvable in polynomial time.

It is useful to note that, for a bipartite graph $V = V_0 \sqcup V_1$, these two problems are equivalent. We can transform the EPR Hamiltonian into the QMC Hamiltonian by rotating the qubits in $V_0$ by $Y$. That is, we transform
\begin{equation} \label{equivalence}
(Y \otimes \mathbbm{1}) g (Y \otimes \mathbbm{1}) = h
\end{equation}
This means that any algorithm for EPR automatically gives an algorithm for QMC on bipartite graphs.

Any algorithm for either problem which searches only over product states has an approximation ratio bounded above by $1/2$. To see this, consider the instance of a single isolated edge $H = h_{12}$. The maximum energy of a product state is $E=1$, yet the entangled singlet state gets energy $E=2$. A similar argument applies to the EPR Hamiltonian.

For a state $|\psi\rangle$, we denote $\langle H\rangle = \langle\psi|H|\psi\rangle$.

\begin{theorem}
There is an efficient classical approximation algorithm for QMC which returns a value $E$ in the range $0.582 \cdot ||H|| \leq E \leq ||H||$ on triangle-free interaction graphs. It outputs a classical description of a state achieving $\langle H \rangle = E$.
\end{theorem}

We suspect that our algorithm achieves the same approximation ratio on all graphs, but our proof only covers the triangle-free case.

\begin{theorem}
There is an efficient classical approximation algorithm for EPR which returns a value $E$ in the range $\frac{1}{\sqrt{2}} ||H|| \leq E \leq ||H||$ on any interaction graphs. It outputs a classical description of a state achieving $\langle H \rangle = E$.
\end{theorem}

The algorithm first relaxes the problem to an efficiently solvable SDP. The SDP is the level-2 Lasserre hierarchy from \cite{Parekh-Thompson}. It then finds an initial good product state, using the rounding algorithm in \cite{Gharibian-Parekh}. In the EPR setting, we can simply take the all-zero state $\bigotimes_V |0\rangle$ as the initial product state. Finally, the algorithm uses the SDP solution to design a quantum circuit of commuting entangling gates. To each edge $(i,j)$, a gate of the form $\exp(i \theta_{ij} P_i \otimes P_j)$ is applied, where $P_i, P_j$ are single-qubit matrices, inspired by \cite{AGM,AGMS}. These gates are designed to rotate the two qubits $(i,j)$ towards the singlet state (or the EPR pair). The output is the entangled state resulting from applying this circuit to the initial product state.

Sections \ref{Lasserre}, \ref{rotation}, \ref{product} elaborate on the three ingredients respectively. Section \ref{algorithm_EPR} discusses a rounding algorithm for EPR, before Section \ref{algorithm_QMC} extends the ideas to construct a rounding algorithm for QMC.

\bigskip
{\noindent \bf \large Notation}

\begin{itemize}
    \item The Pauli matrices are denoted $X = \begin{pmatrix}0&1\\1&0\end{pmatrix}$, $Y = \begin{pmatrix}0&-i\\i&0\end{pmatrix}$, $Z = \begin{pmatrix}1&0\\0&-1\end{pmatrix}$.
    \item If $P$ is a single-qubit operator, $P_i$ denotes $P$ acting on qubit $i$ (for example $X_i$).
    \item Denote the vector of Pauli matrices acting on qubit $i$ by $\vec{\sigma}_i = \begin{pmatrix}X_i\\Y_i\\Z_i\end{pmatrix}$.
    \item When a quantum state $|\psi\rangle$ is clear from the context, we denote $\langle H\rangle = \langle\psi|H|\psi\rangle$.
    \item $\succeq$ denotes that a matrix is positive semi-definite.
\end{itemize}

\section{Prior work}

The first works with SDP rounding algorithms for quantum Hamiltonian problems are \cite{Brandao-Harrow,BGKT}. Brandao and Harrow consider graphs with low threshold rank, and Bravyi et al.\ consider the setting of traceless 2-local Hamiltonians. The outputs of both of these algorithms are product states, or products of few-qubit states.

Gharibian and Parekh develop an algorithm rounding an SDP to a product state in the setting of Quantum Max-Cut \cite{Gharibian-Parekh}. Their state achieves an approximation ratio of $0.498$. In addition \cite{QMC_UGC_hardness} show that this is the optimal approximation algorithm based on the level-1 SDP, assuming a plausible conjecture in Gaussian geometry. This is close to optimal, since product states have an approximation ratio bounded above by $1/2$. \cite{AGM,Parekh-Thompson} are the first to go beyond product state approximation algorithms for Quantum Max-Cut. Anshu, Gosset and Morenz achieve an approximation ratio of $0.531$. Parekh and Thompson introduce the level-2 Lasserre hierarchy in the context of Quantum Max-Cut. They give a rounding algorithm which gets an improved approximation ratio of $0.533$. In \cite{Parekh-Thompson-product}, Parekh and Thompson go back and use the level-2 Lasserre hierarchy to achieve a product state approximation with an exactly optimal approximation ratio of $1/2$.

Another line of work relevant to us is that of \cite{AGM,AGMS}. They study 2-local Hamiltonian problems on low-degree graphs. They show that the energy of any product state can be improved by applying a particular quantum circuit which entangles qubits on the edges of the graph. The circuit we use is heavily inspired by theirs.

As well as improve on the approximation ratio of previous algorithms for Quantum Max-Cut, we make progress in two directions simultaneously.

Firstly, a common feature of the algorithms in \cite{AGM,Parekh-Thompson} is the nature of the entangled states they produce. Namely, their states are a combination of singlets placed on a graph matching, and product states. This is slightly dissatisfying, and we could hope that we could do better with an algorithm outputting a state with more global entanglement. Inspired by \cite{AGM,AGMS}, the algorithm we present rounds the SDP relaxation to a globally entangled state.

Secondly, the construction in \cite{AGM,AGMS} only works on low degree interaction graphs. This is because they rotate each edge by the same amount, and rotating one edge will frustrate adjacent edges. Here we promote their circuit to a variational construction, with a degree of freedom on each edge of the graph. In other words, we allow each edge to be rotated by a different amount. This provides a generalisation of their technique to graphs of arbitrary degree.

An independent work \cite{scoop} by Eunou Lee appearing simultaneously achieves similar results with similar techniques. Here we briefly compare the two works. They achieve an approximation ratio of 0.562 for QMC on \emph{all} graphs. We achieve a slightly better approximation ratio of 0.582, but our analysis only applies to triangle-free graphs. They use the same SDP relaxation as us, namely level-2 of the Lasserre hierarchy; and like us, they round it to the output of a commuting quantum circuit applied to a product state. It also appears that their analysis hinges on the same fact about level-2 Lasserre obeying monogamy of entanglement. This is our Lemma \ref{starbound_lemma}, and Lemma 4 in \cite{scoop}, both originating from Theorem 11 in \cite{Parekh-Thompson}. The main differences lie in the details of the product state and commuting circuit. The initial product state they use is a computational basis state in the Z-basis, whilst our initial product state consists of general vectors on the Bloch sphere. Since the QMC problem is SU(2)-symmetric, we worked hard to ensure our algorithm respects this symmetry. This difference is perhaps responsible for the slightly improved approximation ratio of our work, and also perhaps the difficulty we encountered when analysing graphs with triangles. The formula used to define the circuit rotation angles $\theta_{ij}$ as a function of the SDP differs slightly between the two works. We use a simple $\arcsin$ function, whilst their formula involves the exponential function; see Algorithm 1 in \cite{scoop}.

\section{Level-2 Lasserre hierarchy}\label{Lasserre}

In this section, we introduce the Lasserre hierarchy SDP. We will go through the construction of the SDP in detail. However, the only facts which will be used by the algorithms in later sections are Theorem \ref{upper_bound_thm} and Corollary \ref{starbound_corollary}.

Let $\mathcal{H}_n = (\mathbb{C}^2)^{\otimes n}$ be the Hilbert space of $n$ qubits, and $\mathcal{L}(\mathcal{H}_n)$ the space of operators on $\mathcal{H}_n$. A quantum state is in general given by a density matrix $\rho \in \mathcal{L}(\mathcal{H}_n)$ satisfying
\begin{itemize}
    \item $\rho \succeq 0$
    \item $\Tr{\rho} = 1$
\end{itemize}
From $\rho$ we can derive a bilinear form $M: \mathcal{L}(\mathcal{H}_n) \times \mathcal{L}(\mathcal{H}_n) \rightarrow \mathbb{C}$ given by $M(\alpha,\beta) = \Tr(\rho\alpha^\dag\beta)$. $M$ must satisfy
\begin{itemize}
    \item $M \succeq 0$
    \item $M(\alpha,\beta) = M(\alpha',\beta') \ \forall \ \alpha^\dag \beta = \alpha'^\dag \beta'$
    \item $M(\mathbbm{1},\mathbbm{1}) = 1$
\end{itemize}
In fact, this is an equivalence, in the sense that any bilinear form $M$ satisfying these constraints uniquely specifies a density matrix $\rho$.

Let $H \in \mathcal{H}_n$ be our local Hamiltonian. Now optimizing $\Tr(\rho H)$ over density matrices $\rho$ is equivalent to optimizing $M(\mathbbm{1},H)$ over these positive semi-definite bilinear forms $M$. This already looks like an SDP. The issue is that the dimension of the space $\mathcal{L}(\mathcal{H}_n)$ is exponential in $n$. To deal with this, we \emph{relax} the SDP.

The $d$-local operators span a subspace of operators, call it $\mathcal{L}^{(d)}(\mathcal{H}_n) \subseteq \mathcal{L}(\mathcal{H}_n)$. Note $\mathcal{L}^{(d)}(\mathcal{H}_n) \subseteq \mathcal{L}^{(d')}(\mathcal{H}_n)$ for $d \leq d'$. These subspaces have two crucial properties. Firstly, they contain the 2-local Hamiltonian $H \in \mathcal{L}^{(2)}(\mathcal{H}_n)$. Secondly, the dimension of $\mathcal{L}^{(d)}(\mathcal{H}_n)$ is polynomial in $n$ (albeit exponential in $d$). Thus we can instead optimize $\widetilde{M}(\mathbbm{1},H)$ over the submatrix $\widetilde{M} : \mathcal{L}^{(d)}(\mathcal{H}_n) \times \mathcal{L}^{(d)}(\mathcal{H}_n) \rightarrow \mathbb{C}$ for some fixed locality $d$. We still require this submatrix to have the properties
\begin{itemize}
    \item $\widetilde{M} \succeq 0$
    \item $\widetilde{M}(\alpha,\beta)$ depends only on $\alpha^\dag \beta \in \mathcal{L}(\mathcal{H})$
    \item $\widetilde{M}(\mathbbm{1},\mathbbm{1}) = 1$
\end{itemize}
This will be an efficiently solvable SDP. Notice that $\mathcal{L}^{(d)}(\mathcal{H}_n)$ is \emph{not} closed under multiplication of operators.

This is a relaxation in the sense that any solution to the true problem will give a solution to the relaxation with the same energy. Given a density matrix $\rho$, $\widetilde{M}(\alpha,\beta) := \Tr(\rho\alpha^\dag\beta)$ will be feasible for the relaxed SDP, and we will have $\widetilde{M}(\mathbbm{1},H) = \Tr(\rho H)$. However, given a feasible $\widetilde{M}$, there may not exist a density matrix $\rho$ such that $\widetilde{M}(\alpha,\beta) = \Tr(\rho\alpha^\dag\beta)$. The relaxation can be interpreted as optimizing the Hamiltonian over pseudo-states which are only required to be positive on squares of $d$-local observables, rather than positive on \emph{all} positive semi-definite observables. From here on, we drop the tilde on $\widetilde{M}$.

To write down a concrete SDP relaxation of QMC and EPR, we choose a basis for $\mathcal{L}^{(d)}(\mathcal{H}_n)$. Let $\mathcal{P}^{(d)}_n$ be the set of Pauli operators of weight at most $d$. That is
\begin{equation*}
\mathcal{P}^{(d)}_n = \{\mathbbm{1}, X_1, Y_1, \dots, Y_n, Z_n, X_1 X_2, X_1 Y_2, \dots, (Z_{n-d+1} \dots Z_n)\}
\end{equation*}
These span $\mathcal{L}^{(d)}(\mathcal{H}_n)$ and are orthonormal in the Hilbert-Schmidt inner product.

\begin{definition}\label{SDP}
We are given a weighted graph $(V, \{w_{ij}\})$ with $|V|=n$ vertices. The SDP variable is a matrix $M \in \mathbb{R}^{\mathcal{P}^{(d)}_n \times \mathcal{P}^{(d)}_n}$ with rows and columns indexed by $\mathcal{P}^{(d)}_n$. Define the following SDPs:
\begin{align}
\text{\emph{Lasserre}}_d(H) &= \text{\emph{max}} \ \frac{1}{2} \sum_{ij} w_{ij} \left(1 - M(X_i,X_j) - M(Y_i,Y_j) - M(Z_i,Z_j)\right) \label{SDP_objective_QMC} \\
\text{\emph{Lasserre}}_d(G) &= \text{\emph{max}} \ \frac{1}{2} \sum_{ij} w_{ij} \left(1 + M(X_i,X_j) - M(Y_i,Y_j) + M(Z_i,Z_j)\right) \label{SDP_objective_EPR}
\end{align}
\begin{align}
\text{\emph{s.t.}} \ & M \succeq 0 \nonumber \\
& M(A,B) = M(A',B') \ & \forall \ A^\dag B = A'^\dag B' \nonumber\\
& M(A,B) = -M(A',B') \ & \forall \ A^\dag B = -A'^\dag B' \nonumber \\
& M(A,B) = 0 \ & \forall \ A^\dag B \ \text{\emph{\emph{not} Hermitian}} \nonumber \\
& M(A,A) = 1 \ & \forall \ A \in \mathcal{P}^{(d)}_n \nonumber
\end{align}
\end{definition}

\begin{theorem} \label{upper_bound_thm}
For any constant $d$, $\text{\emph{Lasserre}}_d(H)$ is an efficiently computable semidefinite program which upper bounds \emph{QMC}, and likewise $\text{\emph{Lasserre}}_d(G)$ for \emph{EPR}.
\end{theorem}
\begin{proof}
The statement for $\text{Lasserre}_d(H)$ and QMC is Theorem 7 of \cite{Parekh-Thompson}. The proof for $\text{Lasserre}_d(G)$ and EPR is essentially the same.
\end{proof}

Note that we have taken $M$ to be real symmetric rather than complex Hermitian. This is without loss of generality, since if $M$ is feasible then $\frac{1}{2}(M + M^*)$ is also feasible and achieves the same objective.

We now restrict our attention to the level-2 relaxation $\text{Lasserre}_2$. The task is to use $\text{Lasserre}_2$ to construct a genuine quantum state of high energy. To this end, we extract values for each edge which will be used by the algorithm. For each edge $(i,j)$ define
\begin{align}
1 + x_{ij} &= \frac{1}{2} \left(1 - M(\mathbbm{1},X_iX_j) - M(\mathbbm{1},Y_iY_j) - M(\mathbbm{1},Z_iZ_j)\right) \ , \quad M \text{ solves } \text{Lasserre}_2(H) \label{x_values} \\
1 + y_{ij} &= \frac{1}{2} \left(1 + M(\mathbbm{1},X_iX_j) - M(\mathbbm{1},Y_iY_j) + M(\mathbbm{1},Z_iZ_j)\right) \ , \quad M \text{ solves } \text{Lasserre}_2(G) \label{y_values}
\end{align}

Let's focus on QMC to get some intuition for $x_{ij}$. If $M$ was induced from a valid global density matrix $\rho$, then we would have $1 + x_{ij} = 2\Tr(\psi^-_{ij} \rho) = \Tr(h_{ij} \rho) = \langle h_{ij} \rangle$ the energy of the term $h_{ij}$. Intuitively the values $x_{ij}$ measure how much $\text{Lasserre}_2(H)$ `wants' the edge $(i,j)$ to be close to the singlet state. This is illustrated by the following properties, dropping the subscript $ij$:
\begin{align*}
& -1 \leq x \leq 1 \\
\rho = \psi^- \ &\rightarrow \ x=1 \\
\rho \perp \psi^- \ &\rightarrow \ x=-1 \\
\rho = |01\rangle\langle01| \ &\rightarrow \ x=0
\end{align*}
Similar statements hold for $y_{ij}$ with respect to the EPR state.

The key insight of \cite{Parekh-Thompson} is that $\text{Lasserre}_2$ `knows about' monogamy of entanglement. This can be made precise by the following statement.

\begin{lemma}\label{starbound_lemma}
$\sum_{j \in S} x_{ij} \leq 1$ and $\sum_{j \in S} y_{ij} \leq 1$ for any vertex $i \in V$ and subset $S$ of the neighbours of $i$.
\end{lemma}
\begin{proof}
The statement $\sum_{j \in S} x_{ij} \leq 1$ is a consequence of Theorem 11 of \cite{Parekh-Thompson}. Note that their variables $\{x_e\}$ for $e = (i,j)$ relate to our variables $\{x_{ij}\}$ by $x_e = (2x_{ij} + 1)/3$.

To deduce $\sum_{j \in S} y_{ij} \leq 1$, note that for a star graph, we can transform the EPR Hamiltonian to the QMC Hamiltonian. We do this by rotating the central qubit by $Y$, see Equation \ref{equivalence}.
\end{proof}

Lemma \ref{starbound_lemma} tells us that for any feasible $M$, the edges of a star graph cannot all be simultaneously close to the relevant entangled state. From Lemma \ref{starbound_lemma} we can deduce the following corollary, which will be used when analysing our algorithms in Sections \ref{algorithm_EPR} and \ref{algorithm_QMC}. The proof is in Appendix \ref{starbound_proof}.

\begin{corollary}\label{starbound_corollary}
Fix edge $(i,j)$ and $\beta \in [0,1]$.
\begin{enumerate}
    \item $\prod_{k\sim i : x_{ik} \geq 0} \sqrt{1 - \beta^2 x_{ik}^2} \geq \sqrt{1 - \beta^2}$
    \item $\prod_{j\neq k\sim i : x_{ik} \geq 0} \sqrt{1 - \beta^2 x_{ik}^2} \geq \sqrt{1 - \beta^2 (1-x_{ij})^2}$
\end{enumerate}
and similarly for $y_{ij}$.
\end{corollary}

See Figure \ref{fig_notation} for a description of the notation $j\neq k\sim i$.

\begin{figure}[H]
\centering
  \includegraphics[width=2in]{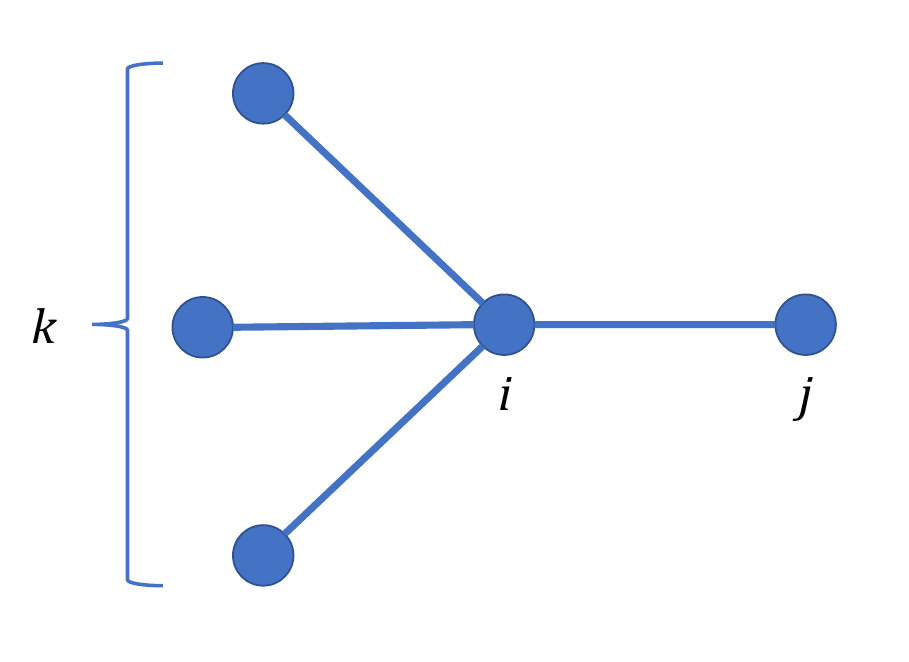}
  \caption{The notation $j\neq k\sim i$ means that $k$ ranges over all vertices that are connected by an edge to $i$, but \emph{not} $j$.}
  \label{fig_notation}
\
\end{figure}

\section{Anshu-Gosset-Morenz quantum circuits}\label{rotation}

We now develop a quantum circuit which improves the energy of a product state, following \cite{AGM,AGMS}. It will consist of 2-qubit unitaries on each edge of the graph, which are designed to rotate the 2-qubit state towards the relevant entangled state. Since it is convenient for these 2-qubit unitaries to commute, they will have the following specific form. Each vertex will be assigned some Hermitian unitary $P_k = \vec{n}_k \cdot \vec{\sigma}_k$ where $\vec{n}_k \in S^2 \subset \mathbb{R}^3$. Then the 2-qubit unitaries will be $\exp(i\theta_{ij} P_i \otimes P_j)$, where $\{\theta_{ij}\}$ are variational parameters associated to each edge. There is a slight abuse of notation here, with the subscript $k$ of $P_k$ denoting that these are \emph{different} matrices, as well as that it is acting on vertex $k\in V$.

To get some intuition for how this unitary acts, consider a single edge $(i,j)$ with the initial product state $|v_i\rangle |v_j\rangle$. Let $\vec{v}_k \in S^2 \subset \mathbb{R}^3$ denote the Bloch vector corresponding to the single-qubit state $|v_k\rangle$, so $|v_k\rangle\langle v_k| = \frac{1}{2} (1 + \vec{v}_k \cdot \vec{\sigma}_k)$.
\begin{equation*}
\exp(i\theta_{ij} P_i \otimes P_j) |v_i\rangle |v_j\rangle = \cos{\theta_{ij}} |v_i\rangle |v_j\rangle + i \sin{\theta_{ij}} P_i|v_i\rangle P_j|v_j\rangle
\end{equation*}

For this to be an entangled state, we would like $P_k |v_k\rangle$ to be orthogonal to $|v_k\rangle$. Thus we would like $\langle v_k|P_k|v_k \rangle = \vec{n}_k \cdot \vec{v}_k$ to equal zero. That is, $\vec{n}_k$ should be orthogonal to $\vec{v}_k$.

\begin{figure}[H]
\centering
  \includegraphics[width=1.5in]{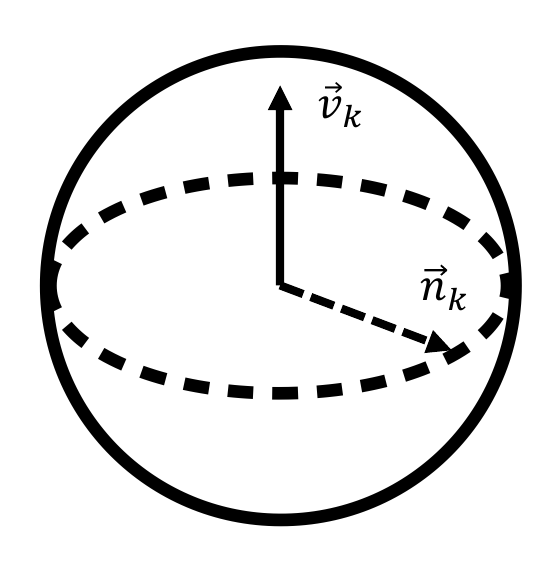}
  \caption{We choose the single-qubit Hermitian unitary $P_k$ to be $P_k = \vec{n}_k \cdot \vec{\sigma}_k$ where $\vec{n}_k$ is a uniformly random unit vector orthogonal to $\vec{v}_k$, $\vec{v}_k \cdot \vec{n}_k = 0$. Geometrically, $\vec{n}_k$ is picked uniformly from the \emph{equator} of the Bloch sphere orthogonal to $\vec{v}_k$.}
\
\end{figure}

Given that this holds, the state becomes
\begin{equation*}
\exp(i\theta_{ij} P_i \otimes P_j) |v_i\rangle |v_j\rangle = \cos{\theta_{ij}} |v_i\rangle |v_j\rangle + e^{i\alpha_{ij}} \sin{\theta_{ij}} |v_i^\perp\rangle |v_j^\perp\rangle
\end{equation*}
where $|v_i^\perp\rangle$ and $|v_j^\perp\rangle$ are single-qubit states orthogonal to $|v_i\rangle$ and $|v_j\rangle$ respectively. $\alpha_{ij}$ is a phase which depends on exactly which $\vec{n}_i$, $\vec{n}_j$ were chosen subject to the orthogonality constraint. Rotating edge $(i,j)$ by $\theta_{ij} = \pi/4$ then corresponds to rotating these two qubits all the way to a maximally entangled 2-qubit state. For us, $\theta_{ij}$ will always be in the range $[0,\pi/4]$.

\subsection{EPR setting}\label{rotation_EPR}

Let's first consider the EPR setting. We want to start with the product state which achieves the highest energy. Without loss of generality, this is the all-zero state $\bigotimes_V |0\rangle$. Then we choose
\begin{equation*}
P_k = \frac{X_k - Y_k}{\sqrt{2}}
\end{equation*}
for all $k$. This indeed has the property $\langle 0|P_k|0 \rangle = 0$. The gate $\exp\left(i \theta_{ij} P_i \otimes P_j \right)$ will be applied to each edge $(i,j)$.
\begin{equation*}
\exp\left(i \theta_{ij} P_i \otimes P_j\right) |00\rangle = \cos{\theta_{ij}} |00\rangle + \sin{\theta_{ij}} |11\rangle
\end{equation*}
As we can see, this unitary rotates $|00\rangle$ towards the EPR state $|\phi^+\rangle = \frac{1}{\sqrt{2}} (|00\rangle + |11\rangle)$. Indeed the choice $\theta_{ij} = \pi/4$ sends $|00\rangle$ precisely to $|\phi^+\rangle$. Letting
\begin{equation*}
\rho = \exp\left(i \theta_{ij} P_i \otimes P_j\right) |00\rangle\langle00| \exp\left(-i \theta_{ij} P_i \otimes P_j\right)
\end{equation*}
we have
\begin{equation} \label{single_edge_intuition_EPR}
\Tr(\phi^+_{ij} \rho) = \frac{1}{2} (1 + \sin{2\theta_{ij}})
\end{equation}
Thus rotating the edge $(i,j)$ will increase the energy of the state on the term $g_{ij}$. However, rotating other edges incident to $i$ or $j$ will have a negative effect on the energy on $g_{ij}$. The following lemma bounds these two competing effects, with the proof in Appendix \ref{rotation_proof_EPR}.

\begin{lemma}\label{rotation_lemma_EPR}
Consider the state $\prod_{kl} \exp\left(\frac{1}{2} i \theta_{kl} (X_k - Y_k) \otimes (X_l - Y_l)\right) \bigotimes_V |0\rangle$. The energy of this state on the Hamiltonian term $g_{ij}$ is given by
\begin{equation}\label{rotation_lemma_EPR_eq}
\langle g_{ij} \rangle \geq \frac{1}{2} + \frac{1}{2} \prod_{j\neq k\sim i} \cos{2\theta_{ik}} \cdot \prod_{i\neq l\sim j} \cos{2\theta_{lj}} + \frac{1}{2} \sin{2\theta_{ij}} \left( \prod_{j\neq k\sim i} \cos{2\theta_{ik}} + \prod_{i\neq l\sim j} \cos{2\theta_{lj}} \right)
\end{equation}
(The notation $j\neq k\sim i$ is shorthand for $k \in \{k \in V : k \neq j , w_{ik} \neq 0\}$.)
\end{lemma}

As we increase the angles $\{\theta_{ij}\}$, the final term in Equation \ref{rotation_lemma_EPR_eq} increases, whilst the middle term decreases. Intuitively, the final term corresponds to the \emph{gain} in energy from the $ij$ rotation, whilst the middle term corresponds to the \emph{loss} in energy from the neighbouring $ik$ and $jl$ rotations.

\subsection{QMC setting}\label{rotation_QMC}

In the QMC setting, the optimal product state is not a fixed state, like the all-zero state $\bigotimes_V |0\rangle$ in the EPR case. Thus we begin with a general product state $\bigotimes_{k\in V} |v_k\rangle$, which we would like to have high energy.

\begin{lemma} \label{Bloch_energy}
Let $\vec{v}_k \in S^2 \subset \mathbb{R}^3$ denote the Bloch vector of the single-qubit state $|v_k\rangle$. The energy of a product state on Hamiltonian term $h_{ij}$ is given by
\begin{equation*}
\langle v_i v_j| h_{ij} |v_i v_j \rangle = 2 \langle v_i v_j| \psi^-_{ij}\rangle \langle\psi^-_{ij}| v_i v_j\rangle = \frac{1}{2} (1 - \vec{v}_i \cdot \vec{v}_j)
\end{equation*}
\end{lemma}
\begin{proof}
Use the representation
\begin{equation*}
|\psi^-_{ij}\rangle \langle\psi^-_{ij}| = \frac{1}{2} (\mathbbm{1} - \text{SWAP}_{ij})
\end{equation*}
where $\text{SWAP}_{ij}$ is the swap operators on qubits $i$ and $j$, and use the fact that
\begin{equation*}
|\langle v_i|v_j\rangle|^2 = \frac{1}{2}(1 + \vec{v}_i \cdot \vec{v}_j)
\end{equation*}
\end{proof}

Define
\begin{equation} \label{E_ij_def}
E_{ij} := \frac{1}{2} (1 - \vec{v}_i \cdot \vec{v}_j)
\end{equation}

For each $k\in V$, independently select $\vec{n}_k \in S^2 \subset \mathbb{R}^3$ uniformly at random from the `equator' of unit vectors satisfying $\vec{n}_k \cdot \vec{v}_k = 0$. Then set $P_k = \vec{n}_k \cdot \vec{\sigma}_k$. Let
\begin{equation*}
\gamma_{ij} = \pi - \arg{\langle v_i|P_j|v_j\rangle\langle v_j|P_i|v_i\rangle}
\end{equation*}

\begin{lemma} \label{equator_lem}
\begin{enumerate}
    \item $\langle v_k|P_k|v_k \rangle = 0$ for each $k$.
    \item The Bloch vector of $P_k|v_k\rangle$ is the \emph{negative} of the Bloch vector of $|v_k\rangle$.
    \item $\gamma_{ij} \in [-\pi,\pi)$ is uniformly random.
\end{enumerate}
\end{lemma}
\begin{proof}
1 and 2 are easy to verify. To see 3, note that $P_i|v_i\rangle$ and $P_j|v_j\rangle$ have independent uniformly random phases.
\end{proof}

In what follows, there is a slight abuse of notation. The subscript $k\in V$ has two roles: it tells us which matrix $P_k$ is, and which qubit it is acting on.
\begin{equation*}
\exp(i\theta_{ij} P_i \otimes P_j) |v_i\rangle |v_j\rangle = \cos{\theta_{ij}} |v_i\rangle |v_j\rangle + i \sin{\theta_{ij}} P_i|v_i\rangle P_j|v_j\rangle
\end{equation*}

Whether or not this state resembles the singlet state $|\psi^-\rangle = \frac{1}{\sqrt{2}} (|01\rangle - |10\rangle)$ depends on the random choices $\vec{n}_i, \vec{n}_j$. Let
\begin{equation*}
\rho = \exp\left(i \theta_{ij} P_i \otimes P_j\right) |v_i v_j\rangle\langle v_i v_j| \exp\left(-i \theta_{ij} P_i \otimes P_j\right)
\end{equation*}
For intuition, let's calculate $\Tr(\psi^-_{ij} \rho)$.
\begin{align*}
\Tr(\psi^-_{ij} \rho) &= \left| \cos{\theta_{ij}} \langle\psi^-_{ij}|v_i v_j\rangle + i \sin{\theta_{ij}} \langle\psi^-_{ij}| P_i \otimes P_j |v_i v_j\rangle \right|^2 \\
&= \cos^2{\theta_{ij}} \langle v_i v_j |\psi^-_{ij}\rangle\langle \psi^-_{ij}|v_i v_j\rangle \\
&\qquad - 2 \sin{\theta_{ij}} \cos{\theta_{ij}} \im{\langle v_i v_j |\psi^-_{ij}\rangle\langle \psi^-_{ij}|P_i \otimes P_j|v_i v_j\rangle} \\
&\qquad + \sin^2{\theta_{ij}} \langle v_i v_j |P_i \otimes P_j|\psi^-_{ij}\rangle\langle \psi^-_{ij}|P_i \otimes P_j|v_i v_j\rangle
\end{align*}
From Lemma \ref{Bloch_energy},
\begin{align*}
\langle v_i v_j |\psi^-_{ij}\rangle\langle \psi^-_{ij}|v_i v_j\rangle &= \frac{1}{2} E_{ij} \\
\langle v_i v_j |P_i \otimes P_j|\psi^-_{ij}\rangle\langle \psi^-_{ij}|P_i \otimes P_j|v_i v_j\rangle &= \frac{1}{2} E_{ij}
\end{align*}
The second equation is true since the angle between the Bloch vectors of $P_i|v_i\rangle$ and $P_j|v_j\rangle$ is the same as that between $|v_i\rangle$ and $|v_j\rangle$.
Using the representation $|\psi^-_{ij}\rangle \langle\psi^-_{ij}| = (\mathbbm{1} - \text{SWAP}_{ij}) / 2$, and recalling $\langle v_i|P_i|v_i\rangle = \langle v_j|P_j|v_j\rangle = 0$, we can write
\begin{align*}
\langle v_i v_j |\psi^-_{ij}\rangle\langle \psi^-_{ij}|P_i \otimes P_j|v_i v_j\rangle &= -\frac{1}{2} \langle v_j|P_i|v_i\rangle\langle v_i|P_j|v_j\rangle \\
&= \frac{1}{2} e^{-i \gamma_{ij}} \left|\langle v_j|P_i|v_i\rangle\langle v_i|P_j|v_j\rangle\right| \\
&= \frac{1}{2} e^{-i \gamma_{ij}} E_{ij}
\end{align*}
The final equality uses Equation \ref{E_ij_def} and that the Bloch vectors of $P_i|v_i\rangle$, $P_j|v_j\rangle$ are the negatives of the Bloch vectors of $|v_i\rangle$, $|v_j\rangle$.
Thus we can complete the calculation as
\begin{equation} \label{single_edge_intuition_QMC}
\Tr(\psi^-_{ij} \rho) = \frac{1}{2} E_{ij} \left(1 + \sin{2\theta_{ij}} \sin{\gamma_{ij}} \right)
\end{equation}
This is reminiscent of Equation \ref{single_edge_intuition_EPR}.

$\gamma_{ij} \in [0,\pi)$ corresponds to rotating $|v_i\rangle |v_j\rangle$ towards the singlet, whilst $\gamma_{ij} \in [-\pi,0)$ in fact corresponds to rotating $|v_i\rangle |v_j\rangle$ \emph{away from} the singlet. To deal with this, we introduce a sign degree of freedom $\varepsilon_{ij} \in \{+1,-1\}$, so that the gates of the circuit are in fact $\exp(i \varepsilon_{ij} \theta_{ij} P_i \otimes P_j)$. This sign can be chosen to ensure we rotate every edge \emph{towards} the singlet as follows.

\begin{equation} \label{sign_choice}
\arg \ \langle v_i|P_j|v_j\rangle\langle v_j|P_i|v_i\rangle =
	\begin{cases}
	[0,\pi) & \rightarrow \ \varepsilon_{ij} = +1 \\
	[-\pi,0) & \rightarrow \ \varepsilon_{ij} = -1
	\end{cases}
\end{equation}

Analogous to Lemma \ref{rotation_lemma_EPR}, there is the following result, with the proof in Appendix \ref{rotation_proof_QMC}. This is the only place in the paper where we use the triangle-free assumption.

\begin{lemma}\label{rotation_lemma_QMC}
Given an initial product state $\bigotimes_{k\in V} |v_k\rangle$, choose single-qubit matrices $\{P_k\}_{k\in V}$ and signs $\{\varepsilon_{kl}\}$ as described above. Consider the state $\prod_{kl} \exp(i \varepsilon_{kl} \theta_{kl} P_k \otimes P_l) \bigotimes_{m \in V} |v_m\rangle$, where the operators range over all edges $k\neq l$ of a triangle-free graph. The expected energy of this state on the Hamiltonian term $h_{ij}$, averaged over the random choices of $\{P_k\}_{k\in V}$, is given by
\begin{equation}\label{rotation_lemma_QMC_eq}
\langle h_{ij} \rangle \geq E_{ij} \left( \frac{1}{2} + \frac{1}{2} \prod_{j\neq k\sim i} \cos{2\theta_{ik}} \cdot \prod_{i\neq l\sim j} \cos{2\theta_{lj}} + \frac{1}{\pi} \sin{2\theta_{ij}} \left( \prod_{j\neq k\sim i} \cos{2\theta_{ik}} + \prod_{i\neq l\sim j} \cos{2\theta_{lj}} \right) \right)
\end{equation}
(The notation $j\neq k\sim i$ is shorthand for $k \in \{k \in V : k \neq j , w_{ik} \neq 0\}$.)
\end{lemma}

As we increase the angles $\{\theta_{ij}\}$, the final term in Equation \ref{rotation_lemma_QMC_eq} increases, whilst the middle term decreases. Intuitively, the final term corresponds to the \emph{gain} in energy from the $ij$ rotation, whilst the middle term corresponds to the \emph{loss} in energy from the neighbouring $ik$ and $jl$ rotations.

Notice that, comparing to Lemma \ref{rotation_lemma_EPR}, the final term is worse by an extra factor of $\frac{2}{\pi}$. This comes from the uniformly random phases $\gamma_{ij} \in [0,2\pi)$.

\section{Gharibian-Parekh product state}\label{product}

In \cite{Gharibian-Parekh}, Gharibian and Parekh give an approximation algorithm for QMC which outputs a product state. Their algorithm (call it GP rounding) rounds the solution of an SDP relaxation to a feasible product state analogously to the Goemans-Williamson hyperplane rounding algorithm.

The SDP they use is in fact the level-1 Lasserre hierarchy $\text{Lasserre}_1(H)$. In this case we must write the objective as
\begin{equation*}\label{level_1_obj}
\text{max} \ \frac{1}{2} \sum_{ij} w_{ij} \left(1 - M(X_i,X_j) - M(Y_i,Y_j) - M(Z_i,Z_j)\right)
\end{equation*}

We will run GP rounding on $\text{Lasserre}_2(H)$. To do this, we first solve $\text{Lasserre}_2(H)$ and then take the submatrix indexed by $\mathcal{P}^{(1)}_n = \{X_1, Y_1, \dots, Y_n, Z_n\}$ to feed into the GP rounding scheme.

GP rounding outputs a product state $\bigotimes_{k \in V} |v_k\rangle$. Let $\vec{v}_k \in S^2 \subset \mathbb{R}^3$ be the corresponding Bloch vectors. Recall Lemma \ref{Bloch_energy}
\begin{equation*}
E_{ij} = \langle v_i| \langle v_j| h_{ij} |v_i \rangle |v_j \rangle = \frac{1}{2} (1 - \vec{v}_i \cdot \vec{v}_j)
\end{equation*}

The following lemma describes the quality of GP rounding relative to the SDP objective. Recall the definition of $\{x_{ij}\}$ from Equation \ref{x_values}, which we can write using Constraint 2 as
\begin{equation*}
1 + x_{ij} = \frac{1}{2} \left(1 - M(X_i,X_j) - M(Y_i,Y_j) - M(Z_i,Z_j)\right)
\end{equation*}

\begin{lemma}
\begin{equation}\label{eq_prod}
E_{ij} \geq f(x_{ij}) := \frac{1}{2} - \frac{4}{3\pi} \left(-\frac{1+2x_{ij}}{3}\right) {}_2F_1\left(1/2,1/2;5/2; \left(-\frac{1+2x_{ij}}{3}\right)^2\right)
\end{equation}
where ${}_2F_1$ is the hypergeometric function defined by
\begin{equation*}
{}_2F_1(a,b;c;z) = \sum_{n=0}^{\infty} \frac{(a)_n (b)_n}{(c)_n} \frac{z^n}{n!} \ , \ (t)_n = \frac{\Gamma(t+n)}{\Gamma(t)} = t(t+1) \dots (t+n-1)
\end{equation*}
Note Equation \ref{eq_prod} defines the function $f$.
\end{lemma}
\begin{proof}
This is implied by Lemma 10 in Appendix C of \cite{Gharibian-Parekh}, which is in turn taken from \cite{BOV12}.
\end{proof}

\begin{figure}[H]
\centering
\begin{subfigure}{.5\textwidth}
  \centering
  \includegraphics[width=3in]{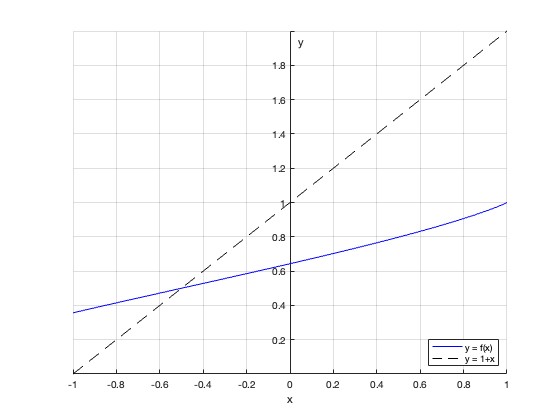}
  \caption{}
  \label{plt_prod_abs}
\end{subfigure}%
\begin{subfigure}{.5\textwidth}
  \centering
  \includegraphics[width=3in]{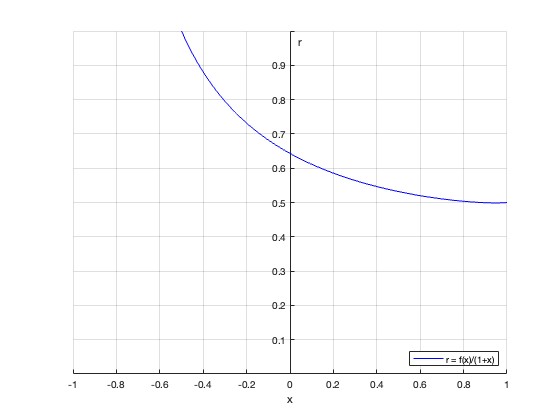}
  \caption{}
  \label{plt_prod_ratio}
\end{subfigure}
\caption{Dropping the subscript $i,j$, it is perhaps instructive to see plots of $f(x)$ versus $1+x$, as well as the ratio $f(x)/(1+x)$. Recall $f(x_{ij})$ is the energy achieved by the Gharibian-Parekh product state on an edge $(i,j)$, whilst the SDP relaxation gives an energy upper bound of $1+x_{ij}$ on that edge. (Given edge $(i,j)$, $x_{ij}$ is defined in the main text.) The minimum of $f(x)/(1+x)$, shown in Figure \ref{plt_prod_ratio}, is $0.498$ at $x = 0.949$. Therefore $0.498$ is the worst-case approximation ratio of the product state on QMC, as shown in \cite{Gharibian-Parekh}. This is almost optimal, in the sense that the approximation ratio of an algorithm outputting a product state cannot exceed $1/2$.}
\label{plt_prod}
\end{figure}

\section{Rounding algorithms}\label{algorithms}

The challenge in designing the algorithms is in choosing the rotation angles $\theta_{ij}$ for the circuit from Section \ref{rotation}. For low degree graphs there is not much frustration, so it should be possible to rotate almost all the way to $\theta_{ij} = \pi/4$. On the other hand for high degree graphs, product states are expected to provide good approximations to the ground state, so in this case the angles $\theta_{ij}$ should be small. To make these choices, we consult the values $x_{ij}$ or $y_{ij}$ derived from $\text{Lasserre}_2$ in Section \ref{Lasserre}. $x_{ij}$ and $y_{ij}$ measure how much entanglement $\text{Lasserre}_2$ wants to place on the edge $(i,j)$. Crucially, Lemma \ref{starbound_lemma} tells us that $\text{Lasserre}_2$ `knows about' monogamy of entanglement.

\subsection{Algorithm for EPR}\label{algorithm_EPR}

We first give our algorithm for the EPR Hamiltonian problem on any graph.

\begin{algorithm}\label{alg_EPR}
Input: weighted graph $(V, \{w_{ij}\})$ with $\sum_{ij} w_{ij} = 1$.
\begin{enumerate}
	\item Solve $\text{Lasserre}_2(G)$ from Definition \ref{SDP} to get solution $M \in \mathbb{R}^{\mathcal{P}^{(2)}_n \times \mathcal{P}^{(2)}_n}$.
	\item For each edge $(i,j)$ compute $y_{ij} = \frac{1}{2} \left(-1 + M(\mathbbm{1},X_iX_j) - M(\mathbbm{1},Y_iY_j) + M(\mathbbm{1},Z_iZ_j)\right)$.
	\item Set $\theta_{ij} =
	\begin{cases}
	\frac{1}{2} \arcsin{y_{ij}} & y_{ij} \geq 0 \\
	0 & \text{\emph{otherwise}}
	\end{cases}$
	\item Compute $\eta = \sum_{ij} w_{ij} y_{ij}$.
	\item If $\eta \geq \sqrt{2}-1$, return the state $\prod_{ij} \exp\left(\frac{1}{2} i \theta_{ij} (X_i - Y_i) \otimes (X_j - Y_j)\right) \bigotimes_V |0\rangle$ with energy $E \geq \frac{1}{2} + \eta + \frac{1}{2} \eta^2$. Else if $\eta < \sqrt{2}-1$, return the state $\bigotimes_V |0\rangle$ with energy $E = 1$.
\end{enumerate}
\end{algorithm}

The first thing the EPR algorithm does is formulate and solve $\text{Lasserre}_2(G)$ from Section \ref{Lasserre}. It then computes the values $y_{ij}$ for each edge $(i,j)$ defined in Equation \ref{y_values}, and sets
\begin{equation}
\theta_{ij} =
	\begin{cases}
	\frac{1}{2} \arcsin{y_{ij}} & y_{ij} \geq 0 \\
	0 & \text{otherwise}
	\end{cases}
\end{equation}
To motivate this assignment, notice that for a single edge
\begin{equation*}
\Tr \left(g_{ij} \exp\left(\frac{1}{2} i \theta_{ij} (X_i - Y_i) \otimes (X_j - Y_j)\right) |00\rangle\langle00| \exp\left(-\frac{1}{2} i \theta_{ij} (X_i - Y_i) \otimes (X_j - Y_j)\right)\right) = 1 + \sin{2\theta_{ij}}
\end{equation*}
From $\text{Lasserre}_2(G)$ we have $\Tr(g_{ij} \tilde{\rho}) = 1 + y_{ij}$ for the pseudo-density $\tilde{\rho}$. We have chosen $\theta_{ij}$ so that this matches Equation \ref{single_edge_intuition_EPR} for a single edge.

The algorithm then applies the circuit of commuting gates $\prod_{ij} \exp\left(\frac{1}{2} i \theta_{ij} (X_i - Y_i) \otimes (X_j - Y_j)\right)$ to the product state $\bigotimes_V |0\rangle$ as described in Section \ref{rotation_EPR}.

The energy of the rotated state $\prod_{ij} \exp\left(\frac{1}{2} i \theta_{ij} (X_i - Y_i) \otimes (X_j - Y_j)\right) \bigotimes_V |0\rangle$ can be computed from Lemma \ref{rotation_lemma_EPR} and Corollary \ref{starbound_corollary}.

\begin{lemma}
\begin{equation}
\langle g_{ij} \rangle \geq \frac{1}{2} + y_{ij} + \frac{1}{2} y_{ij}^2
\end{equation}
\end{lemma}
\begin{proof}
First suppose $y_{ij} \geq 0$. Using Lemma \ref{rotation_lemma_EPR},
\begin{align*}
\langle g_{ij} \rangle &\geq \frac{1}{2} + \frac{1}{2} \prod_{j\neq k\sim i : y_{ik} \geq 0} \sqrt{1-y_{ik}^2} \cdot \prod_{i\neq l\sim j : y_{lj} \geq 0} \sqrt{1-y_{lj}^2} \\
&\quad + \frac{1}{2} y_{ij} \left(\prod_{j\neq k\sim i : y_{ik} \geq 0} \sqrt{1-y_{ik}^2} + \prod_{i\neq l\sim j : y_{lj} \geq 0} \sqrt{1-y_{lj}^2}\right) \\
&\geq \frac{1}{2} + y_{ij} + \frac{1}{2} y_{ij}^2
\end{align*}
where we used Corollary \ref{starbound_corollary} with $\beta = 1$ to write
\begin{equation*}
\prod_{j\neq k\sim i : y_{ik} \geq 0} \sqrt{1-y_{ik}^2} + \prod_{i\neq l\sim j : y_{lj} \geq 0} \sqrt{1-y_{lj}^2} \ \geq \ 2\sqrt{1 - (1-y_{ij})^2} \ \geq \ 2y_{ij}
\end{equation*}
\begin{equation*}
\prod_{j\neq k\sim i : y_{ik} \geq 0} \sqrt{1-y_{ik}^2} \cdot \prod_{i\neq l\sim j : y_{lj} \geq 0} \sqrt{1-y_{lj}^2} \ \geq \ 1 - (1-y_{ij})^2 \ = \ 2y_{ij} - y_{ij}^2
\end{equation*}

Now if $y_{ij} < 0$ then we also have $\langle g_{ij} \rangle \geq \frac{1}{2} \geq \frac{1}{2} + y_{ij} + \frac{1}{2} y_{ij}^2$.
\end{proof}

Thus the total energy of the rotated state is
\begin{equation*}
\langle G \rangle \geq \sum_{ij} w_{ij} \left(\frac{1}{2} + y_{ij} + \frac{1}{2} y_{ij}^2\right)
\end{equation*}

To help analyse the algorithm, let's interpret $w_{ij}$ as a probability measure on the edges of the graph. We can then write
\begin{align*}
\langle G \rangle &\geq \mathbb{E} \left(\frac{1}{2} + y_{ij} + \frac{1}{2} y_{ij}^2\right) \\
&\geq \frac{1}{2} + \mathbb{E} y_{ij} + \frac{1}{2} (\mathbb{E} y_{ij})^2 \\
&= \frac{1}{2} + \eta + \frac{1}{2} \eta^2
\end{align*}
where we have used Jensen's inequality and defined
\begin{equation*}
\eta := \mathbb{E} y_{ij} = \sum_{ij} w_{ij} y_{ij}
\end{equation*}
Note that the use of Jensen's inequality relies on the fact that $\frac{1}{2} + y + \frac{1}{2} y^2$ is a convex function in $y$.

The rotated state does well only if the $y_{ij}$ values are large i.e. close to 1. This corresponds to the case of low degree and low frustration. If the $y_{ij}$ values are small or negative, it is possible that the rotated state will have energy smaller than that of the original product state $\bigotimes_V |0\rangle$, which achieves $\langle G \rangle = 1$. This is the high degree and high frustration case. In this case, the algorithm can simply use the original product state.

We remark that the feature that our algorithm is completely deterministic.

\begin{theorem} \label{EPR_alg_thm}
The output energy of Algorithm \ref{alg_EPR} satisfies $\frac{1}{\sqrt{2}} ||G|| \leq E \leq ||G||$.
\end{theorem}
\begin{proof}
Recall that $\text{Lasserre}_2(G)$ obtains objective $1 + \sum_{ij} w_{ij} y_{ij} = 1+\eta$, which is an upper bound on $||G||$. Thus the approximation ratio is
\begin{align*}
\frac{E}{||G||} &\geq \text{min}_\eta \ \text{max} \left(\frac{1}{1+\eta} \ , \ \frac{\frac{1}{2}+\eta+\frac{1}{2}\eta^2}{1+\eta}\right) \\
&= \text{min}_\eta \ \text{max} \left(\frac{1}{1+\eta} \ , \ \frac{1}{2} (1 + \eta)\right) \\
&= \frac{1}{\sqrt{2}}
\end{align*}
The crossover between the two ratios occurs at $\eta = \sqrt{2} - 1$.
\end{proof}

\begin{figure}[H]
\centering
  \includegraphics[width=3in]{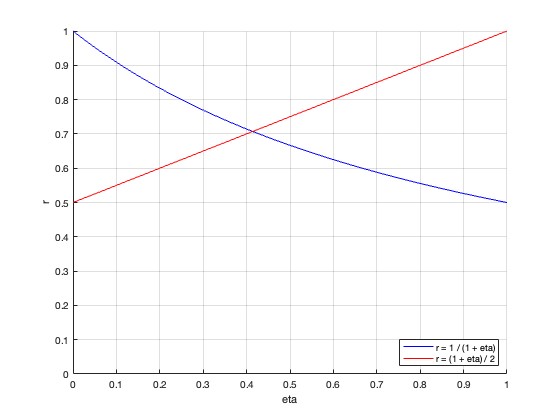}
\caption{It is instructive to see plots of the functions involved in the proof of Theorem \ref{EPR_alg_thm}. Given $\eta$ as defined in the main text, the product state achieves approximation ratio $1/(1+\eta)$, and the output of the commuting circuit achieves approximation ratio $(1+\eta)/2$. Algorithm \ref{alg_EPR} takes whichever state is better. The crossover occurs at $\eta = \sqrt{2} - 1$, where $1/(1+\eta) = (1+\eta)/2 = 1/\sqrt{2}$; thus the final worst-case approximation ratio of Algorithm \ref{alg_EPR} is $1/\sqrt{2}$.}
\label{plt_ratio_EPR}
\
\end{figure}

\subsection{Algorithm for QMC}\label{algorithm_QMC}

We now move onto our algorithm for the full Quantum Max-Cut problem on triangle-free graphs.

\begin{algorithm}\label{alg_QMC}
Input: triangle-free weighted graph $(V, \{w_{ij}\})$ with $\sum_{ij} w_{ij} = 1$.
\begin{enumerate}
	\item Solve $\text{Lasserre}_2(H)$ from Definition \ref{SDP} to get solution $M \in \mathbb{R}^{\mathcal{P}^{(2)}_n \times \mathcal{P}^{(2)}_n}$.
	\item Perform GP rounding \cite{Gharibian-Parekh} on $\text{Lasserre}_2(H)$ to get product state $\bigotimes_{k\in V} |v_k\rangle$.
	\item For each edge $(i,j)$ compute $x_{ij} = - \frac{1}{2} \left(1 + M(\mathbbm{1},X_iX_j) + M(\mathbbm{1},Y_iY_j) + M(\mathbbm{1},Z_iZ_j)\right)$.
	\item Let $F(\beta;x) = f(x) \left(\frac{1}{2} + \frac{1}{2}\left(1 - \beta^2 (1-x)^2\right) + \frac{2}{\pi} \beta x \left(\sqrt{1-\beta^2} + \left(1 - \sqrt{1-\beta^2}\right) x\right)\right)$, where $\beta \in [0,1]$ and $x \in [-1,1]$.
	\item Set $\beta^* = \text{argmax}_\beta \ \text{min}_x \ \frac{F(\beta;x)}{1+x} \approx 0.390$.
	\item Set $\theta_{ij} =
	\begin{cases}
	\frac{1}{2} \arcsin{\beta^* x_{ij}} & x_{ij} \geq 0 \\
	0 & \text{\emph{otherwise}}
	\end{cases}$
	\item For each $k\in V$, independently select $\vec{n}_k \in S^2 \subset \mathbb{R}^3$ uniformly at random from the unit vectors satisfying $\vec{n}_k \cdot \vec{v}_k = 0$, where $\vec{v}_k \in S^2$ is the Bloch vector of $|v_k\rangle$.
	\item Set $P_k = \vec{n}_k \cdot \vec{\sigma}_k$.
	\item Choose the signs $\varepsilon_{ij} \in \{+1,-1\}$ according to $\arg \langle v_i|P_j|v_j\rangle\langle v_j|P_i|v_i\rangle =
	\begin{cases}
	[0,\pi) & \varepsilon_{ij} = +1 \\
	[-\pi,0) & \varepsilon_{ij} = -1
	\end{cases}$
	\item Return the state $\prod_{ij} \exp(i \varepsilon_{ij} \theta_{ij} P_i \otimes P_j) \bigotimes_{k\in V} |v_k\rangle$ with expected energy $E \geq 0.582\cdot \text{\emph{Lasserre}}_2(H) \geq 0.582\cdot ||H||$.
\end{enumerate}
\end{algorithm}

The algorithm has three components. First, it formulates and solves $\text{Lasserre}_2(H)$ from Section \ref{Lasserre}. In the second stage, it performs GP rounding \cite{Gharibian-Parekh} on the SDP solution to derive the product state $\bigotimes_{k\in V} |v_k\rangle$ with the properties detailed in Section \ref{product}.

From $\text{Lasserre}_2(H)$ it computes the values $x_{ij}$ for each edge $(i,j)$ defined in Equation \ref{x_values}. It then sets
\begin{equation}
\theta_{ij} =
	\begin{cases}
	\frac{1}{2} \arcsin{\beta^* x_{ij}} & x_{ij} \geq 0 \\
	0 & \text{otherwise}
	\end{cases}
\end{equation}

The parameter $\beta$ was not present in the EPR algorithm. Although setting $\beta=1$ is natural, we found that optimizing $\beta \in [0,1]$ is useful. Firstly, it improves the approximation ratio. Secondly, a common feature of entangled approximation algorithms is that they choose between an entangled state and a product state at the end. Indeed, this was the case for the EPR algorithm in Section \ref{algorithm_EPR}. The introduction of $\beta$ allows a unified algorithm to achieve the approximation ratio alone, without needing to take the maximum of two separate algorithms as in the EPR case. It is perhaps possible to introduce an analogous parameter into the EPR algorithm; we kept $\beta=1$ in Section \ref{algorithm_EPR} for the sake of simplicity and clarity.

The final component is to perform the circuit of commuting gates $\prod_{ij} \exp(i \varepsilon_{ij} \theta_{ij} P_i \otimes P_j)$ on $\bigotimes_{k\in V} |v_k\rangle$ as described in Section \ref{rotation_QMC}. This involves using knowledge of the states $|v_k\rangle$ to randomly choose the single-qubit matrices $P_k$, and given a choice of $P_k$, choosing the signs of the rotation angles $\varepsilon_{ij} \in \{+1,-1\}$.

The expected energy of the rotated state $\prod_{ij} \exp(i \varepsilon_{ij} \theta_{ij} P_i \otimes P_j) \bigotimes_{k \in K} |v_k\rangle$, averaged over the random choices of $P_k$, can be computed from Lemma \ref{rotation_lemma_QMC} and Corollary \ref{starbound_corollary}.

\begin{lemma}
Defining
\begin{equation}
\theta_{ij} =
	\begin{cases}
	\frac{1}{2} \arcsin{\beta x_{ij}} & x_{ij} \geq 0 \\
	0 & \text{otherwise}
	\end{cases}
\end{equation}
with variable $\beta$, we have
\begin{equation}
\langle h_{ij} \rangle \geq E_{ij} \left(\frac{1}{2} + \frac{1}{2}\left(1 - \beta^2 (1-x_{ij})^2\right) + \frac{2}{\pi} \beta x_{ij} \left(\sqrt{1-\beta^2} + \left(1 - \sqrt{1-\beta^2}\right) x_{ij}\right)\right) \label{beta_lower_bound}
\end{equation}
\end{lemma}
\begin{proof}
First suppose $x_{ij} \geq 0$. Using Lemma \ref{rotation_lemma_QMC},
\begin{align*}
\langle h_{ij} \rangle &\geq E_{ij} \Bigg( \frac{1}{2} + \frac{1}{2} \prod_{j\neq k\sim i : x_{ik} \geq 0} \sqrt{1 - \beta^2 x_{ik}^2} \cdot \prod_{i\neq l\sim j : x_{lj} \geq 0} \sqrt{1 - \beta^2 x_{lj}^2} \\
&\quad + \frac{1}{\pi} \beta x_{ij} \Bigg(\prod_{j\neq k\sim i : x_{ik} \geq 0} \sqrt{1 - \beta^2 x_{ik}^2} + \prod_{i\neq l\sim j : x_{lj} \geq 0} \sqrt{1 - \beta^2 x_{lj}^2}\Bigg) \Bigg)
\end{align*}

We can use Corollary \ref{starbound_corollary} with $\beta = 1$ to write
\begin{align*}
\prod_{j\neq k\sim i : x_{ik} \geq 0} \sqrt{1 - \beta^2 x_{ik}^2} + \prod_{i\neq l\sim j : x_{lj} \geq 0} \sqrt{1 - \beta^2 x_{lj}^2} \ &\geq \ 2\sqrt{1 - \beta^2 (1-x_{ij})^2} \\ &\geq \ 2 \left(\sqrt{1 - \beta^2} + \left(1 - \sqrt{1-\beta^2}\right) x_{ij}\right)
\end{align*}
The second inequality uses the fact that the function $\sqrt{1 - \beta^2 (1-x)^2}$ is concave on $x \in [0,1]$. We have lower bounded it by the linear function interpolating between $(0,\sqrt{1-\beta^2})$ and $(1,1)$.

We can also use Corollary \ref{starbound_corollary} to write
\begin{equation*}
\prod_{j\neq k\sim i : x_{ik} \geq 0} \sqrt{1 - \beta^2 x_{ik}^2} \cdot \prod_{i\neq l\sim j : x_{lj} \geq 0} \sqrt{1 - \beta^2 x_{lj}^2} \ \geq \ 1 - \beta^2 (1-x_{ij})^2
\end{equation*}

These give
\begin{equation*}
\langle h_{ij} \rangle \geq E_{ij} \left(\frac{1}{2} + \frac{1}{2}\left(1 - \beta^2 (1-x_{ij})^2\right) + \frac{2}{\pi} \beta x_{ij} \left(\sqrt{1-\beta^2} + \left(1 - \sqrt{1-\beta^2}\right) x_{ij}\right)\right)
\end{equation*}

Now if $x_{ij} \leq 0$ then
\begin{align*}
\langle h_{ij} \rangle &\geq E_{ij} \left( \frac{1}{2} + \frac{1}{2} \prod_{j\neq k\sim i : x_{ik} \geq 0} \sqrt{1 - \beta^2 x_{ik}^2} \cdot \prod_{i\neq l\sim j : x_{lj} \geq 0} \sqrt{1 - \beta^2 x_{lj}^2} \right) \\
&\geq E_{ij} \left(\frac{1}{2} + \frac{1}{2}(1 - \beta^2)\right)
\end{align*}
using Corollary \ref{starbound_corollary}. This is greater than the expression in Equation \ref{beta_lower_bound} if $x_{ij} \leq 0$, so Equation \ref{beta_lower_bound} holds in both cases $x_{ij} \geq 0$ and $x_{ij} \leq 0$.
\end{proof}

Since $E_{ij} \geq f(x_{ij})$, we have
\begin{align*}
\langle h_{ij} \rangle &\geq f(x_{ij}) \left(\frac{1}{2} + \frac{1}{2}\left(1 - \beta^2 (1-x_{ij})^2\right) + \frac{2}{\pi} \beta x_{ij} \left(\sqrt{1-\beta^2} + \left(1 - \sqrt{1-\beta^2}\right) x_{ij}\right)\right)\\
&=: F(\beta;x_{ij})
\end{align*}

We would like $\langle h_{ij}\rangle$ to be as high as possible relative to the value achieved by $\text{Lasserre}_2(H)$ on edge $ij$, which is $1+x_{ij}$. This leads us to choose $\beta$ to be the value maximising $\text{min}_x \frac{F(\beta,x)}{1+x}$. Denote this $\beta^* := \text{argmax}_\beta \text{min}_x \frac{F(\beta,x)}{1+x} \approx 0.390$. This gives
\begin{equation*}
\text{min}_x \frac{F(\beta^*,x)}{1+x} \approx 0.5828
\end{equation*}

\begin{theorem}
The output energy of Algorithm \ref{alg_QMC} satisfies $0.582 \cdot ||H|| \leq E \leq ||H||$.
\end{theorem}
\begin{proof}
\begin{align*}
E &\geq \sum_{ij} w_{ij} F(\beta^*;x_{ij}) \\
&\geq 0.582\cdot \sum_{ij} w_{ij} (1 + x_{ij}) \\
&= 0.582\cdot \text{Lasserre}_2(H) \\
&\geq 0.582\cdot ||H||
\end{align*}
\end{proof}

A nice feature of this algorithm is that it maintains the $SU(2)$ rotation symmetry of the problem. For example, nowhere does the algorithm pick a preferred single-qubit basis, or a preferred direction on the Bloch sphere. This is in contrast to other QMC algorithms such as the one in \cite{AGM, Parekh-Thompson, scoop}, which often pick a preferred basis $\{|0\rangle,|1\rangle\}$.

\begin{figure}[H]
\centering
\begin{subfigure}{.5\textwidth}
  \centering
  \includegraphics[width=3in]{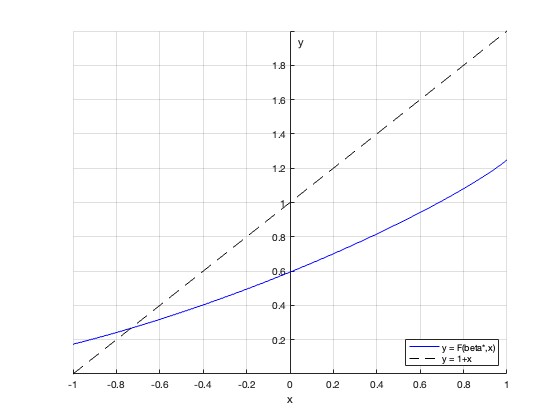}
  \caption{}
  \label{plt_QMC_abs}
\end{subfigure}%
\begin{subfigure}{.5\textwidth}
  \centering
  \includegraphics[width=3in]{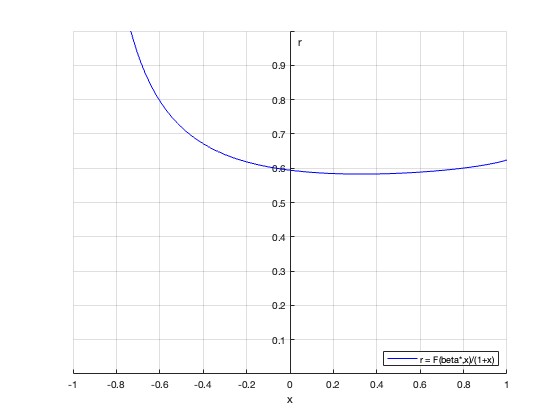}
  \caption{}
  \label{plt_QMC_ratio}
\end{subfigure}
\caption{Dropping the subscript $i,j$, the figure shows plots of the functions $F(\beta^*,x)$ versus $1+x$ and their ratio $F(\beta^*,x) / (1+x)$. Recall $F(\beta^*,x_{ij})$ is the energy achieved by Algorithm \ref{alg_QMC} on an edge $(i,j)$, whilst the SDP relaxation gives an energy upper bound of $1+x_{ij}$ on that edge. (Given edge $(i,j)$, $x_{ij}$ is defined in the main text.) The minimum of $F(\beta^*,x) / (1+x)$, shown in Figure \ref{plt_QMC_ratio}, is $0.582$ at $x = 0.329$. Therefore $0.582$ is the worst-case approximation ratio of Algorithm \ref{alg_QMC}.}
\label{plt_QMC}
\end{figure}

\section{Future directions}

We believe it would be fruitful to search for better approximation algorithms to the EPR Hamiltonian. Since optimizing the EPR Hamiltonian over product states is trivial, this problem serves to isolate the quantum task of optimizing the landscape of entanglement in the groundstate. An approximation algorithm would need to deal with the constraint of monogamy of entanglement.

It is very possible that there are better algorithms out there for Quantum Max-Cut. One way to imagine improving Algorithm \ref{alg_QMC} is to choose the matrices $P_k = \vec{n}_k \cdot \vec{\sigma}_k$ in a manner more sophisticated than independently and uniformly at random. Improvements to both Algorithms \ref{alg_EPR} and \ref{alg_QMC} may be possible by a better choice of rotation angles $\theta_{ij}$ as a function of $\text{Lasserre}_2$. Incidentally, we have no reason to believe that our analyses of the algorithms here are tight. Additionally, an average case analysis of algorithms for Quantum Max-Cut would be an interesting direction.

The Goemans-Williamson algorithm solves classical Max-Cut with an approximation ratio of 0.878 \cite{GW}. It works by formulating an SDP, and rounding it to a classical solution via a randomized hyperplane. Remarkably, this is optimal in two different senses.

Firstly, the SDP has an integrality gap exactly matching the approximation ratio of Goemans-Williamson \cite{GW_integrality_gap}. This is an instance $K$ of classical Max-Cut with $\text{MaxCut}(K) = 0.878\cdot\text{SDP}(K)$. This shows that the hyperplane rounding algorithm is optimal.

Can we establish intregrality gaps for $\text{Lasserre}_2$ on QMC and EPR? This would be an instance $H$ of Quantum Max-Cut where $||H|| = r\cdot\text{Lasserre}_2(H)$ for some $r<1$, and similarly for the EPR Hamiltonian. They would serve as upper bounds on the performance of algorithms based on $\text{Lasserre}_2$.

The second way in which Goemans-Williamson is optimal is a conditional hardness of approximation result. Assuming the Unique Games Conjecture, it is impossible to approximate classical Max-Cut to within a factor better than 0.878 \cite{GW_UGC_hardness}. Assuming the weaker $\text{P} \neq \text{NP}$, it is impossible to approximate within factor $0.941$ \cite{GW_NP_hardness}.

We can similarly ask for conditional hardness results for Quantum Max-Cut. In \cite{QMC_UGC_hardness}, the authors show Unique Games hardness of approximation for QMC to within a factor of 0.956, assuming a plausible conjecture in Gaussian geometry. Can we show NP-hardness of approximation? And can these approximation ratios be made tighter?

Even more interesting would be QMA-hardness of approximation. Unconditional QMA-hardness of approximation to within some constant factor would constitute a quantum PCP theorem. Can we make a QMA-hardness conjecture, analogous to the Unique Games Conjecture in the classical setting, under which we can show QMA-hardness of approximation for Quantum Max-Cut?

The questions above are motivated from existing knowledge on approximation algorithms for classical NP-hard problems. However, the answers may turn out to be much more subtle than in the classical setting. Far from being discouraging, we hope that this will provide insight into the ways in which quantum Hamiltonian problems differ from classical constraint satisfaction.

Last but not least, the high level idea of our algorithms is to use an SDP relaxation to inform the design of a quantum circuit. The SDP allows us to classically efficiently choose good parameters for a family of variational quantum circuits. It would be very interesting to see if this approach can be applied to other problems in quantum information.

\bigskip
{\noindent\bf\large Acknowledgements}

The author is very grateful to John Wright for in-depth discussions and detailed feedback on the write-up. The author also thanks Chaithanya Rayudu for a useful observation, Sevag Gharibian for helpful discussions, and Thomas Vidick for feedback on the write-up.

\bibliographystyle{alpha}
\bibliography{refs}

\newcommand{\etalchar}[1]{$^{#1}$}
\begin{thebibliography}{BdOFV14}

\bibitem[AGM20]{AGM}
Anurag Anshu, David Gosset, and Karen Morenz.
\newblock Beyond product state approximations for a quantum analogue of max
  cut.
\newblock In {\em 15th Conference on the Theory of Quantum Computation,
  Communication and Cryptography}, 2020.
\newblock \url{https://doi.org/10.4230/LIPIcs.TQC.2020.7}.

\bibitem[AGMS21]{AGMS}
Anurag Anshu, David Gosset, Karen Morenz, and Mehdi Soleimanifar.
\newblock Improved approximation algorithms for bounded-degree local
  hamiltonians.
\newblock {\em Physical Review Letters}, 127(25):250502, 2021.
\newblock \url{https://doi.org/10.1103/PhysRevLett.127.250502}.

\bibitem[BdOFV14]{BOV12}
Jop Bri{\"e}t, Fernando~M{\'a}rio de~Oliveira~Filho, and Frank Vallentin.
\newblock Grothendieck inequalities for semidefinite programs with rank
  constraint.
\newblock {\em Theory OF Computing}, 10(4):77--105, 2014.
\newblock \url{https://doi.org/10.4086/toc.2014.v010a004}.

\bibitem[BGKT19]{BGKT}
Sergey Bravyi, David Gosset, Robert K{\"o}nig, and Kristan Temme.
\newblock Approximation algorithms for quantum many-body problems.
\newblock {\em Journal of Mathematical Physics}, 60(3):032203, 2019.
\newblock \url{https://doi.org/10.1063/1.5085428}.

\bibitem[BH16]{Brandao-Harrow}
Fernando~GSL Brandao and Aram~W Harrow.
\newblock Product-state approximations to quantum states.
\newblock {\em Communications in Mathematical Physics}, 342:47--80, 2016.
\newblock \url{https://doi.org/10.1007/s00220-016-2575-1}.

\bibitem[CM16]{QMA_complete_2}
Toby Cubitt and Ashley Montanaro.
\newblock Complexity classification of local hamiltonian problems.
\newblock {\em SIAM Journal on Computing}, 45(2):268--316, 2016.
\newblock \url{https://doi.org/10.1137/140998287}.

\bibitem[FS02]{GW_integrality_gap}
Uriel Feige and Gideon Schechtman.
\newblock On the optimality of the random hyperplane rounding technique for max
  cut.
\newblock {\em Random Structures \& Algorithms}, 20(3):403--440, 2002.
\newblock \url{https://doi.org/10.1002/rsa.10036}.

\bibitem[GK12]{QMA_approximation}
Sevag Gharibian and Julia Kempe.
\newblock Approximation algorithms for qma-complete problems.
\newblock {\em SIAM Journal on Computing}, 41(4):1028--1050, 2012.
\newblock \url{https://doi.org/10.1137/110842272}.

\bibitem[GP19]{Gharibian-Parekh}
Sevag Gharibian and Ojas Parekh.
\newblock Almost optimal classical approximation algorithms for a quantum
  generalization of max-cut.
\newblock In {\em Approximation, Randomization, and Combinatorial Optimization.
  Algorithms and Techniques (APPROX/RANDOM 2019)}. Schloss
  Dagstuhl-Leibniz-Zentrum fuer Informatik, 2019.
\newblock \url{https://doi.org/10.4230/LIPIcs.APPROX-RANDOM.2019.31}.

\bibitem[GW95]{GW}
Michel Goemans and David Williamson.
\newblock Improved approximation algorithms for maximum cut and satisfiability
  problems using semidefinite programming.
\newblock {\em Journal of the ACM (JACM)}, 42(6):1115--1145, 1995.
\newblock \url{https://doi.org/10.1145/227683.227684}.

\bibitem[H{\aa}s01]{GW_NP_hardness}
Johan H{\aa}stad.
\newblock Some optimal inapproximability results.
\newblock {\em Journal of the ACM (JACM)}, 48(4):798--859, 2001.
\newblock \url{https://doi.org/10.1145/502090.502098}.

\bibitem[HNP{\etalchar{+}}23]{QMC_UGC_hardness}
Yeongwoo Hwang, Joe Neeman, Ojas Parekh, Kevin Thompson, and John Wright.
\newblock Unique games hardness of quantum max-cut, and a conjectured
  vector-valued borell's inequality.
\newblock In {\em Proceedings of the 2023 Annual ACM-SIAM Symposium on Discrete
  Algorithms (SODA)}, pages 1319--1384. SIAM, 2023.
\newblock \url{https://doi.org/10.1137/1.9781611977554.ch48}.

\bibitem[KKMO07]{GW_UGC_hardness}
Subhash Khot, Guy Kindler, Elchanan Mossel, and Ryan O’Donnell.
\newblock Optimal inapproximability results for max-cut and other 2-variable
  csps?
\newblock {\em SIAM Journal on Computing}, 37(1):319--357, 2007.
\newblock \url{https://doi.org/10.1137/S0097539705447372}.

\bibitem[Lee22]{scoop}
Eunou Lee.
\newblock Optimizing quantum circuit parameters via sdp.
\newblock In {\em 33rd International Symposium on Algorithms and Computation
  (ISAAC 2022)}. Schloss Dagstuhl-Leibniz-Zentrum f{\"u}r Informatik, 2022.
\newblock \url{https://doi.org/10.4230/LIPIcs.ISAAC.2022.48}.

\bibitem[PM17]{QMA_complete}
Stephen Piddock and Ashley Montanaro.
\newblock The complexity of antiferromagnetic interactions and 2d lattices.
\newblock {\em Quantum Information \& Computation}, 17(7-8):636--672, 2017.
\newblock \url{https://dl.acm.org/doi/abs/10.5555/3179553.3179559}.

\bibitem[PT21a]{Parekh-Thompson}
Ojas Parekh and Kevin Thompson.
\newblock Application of the level-2 quantum lasserre hierarchy in quantum
  approximation algorithms.
\newblock In {\em 48th International Colloquium on Automata, Languages, and
  Programming (ICALP 2021)}. Schloss Dagstuhl-Leibniz-Zentrum f{\"u}r
  Informatik, 2021.
\newblock \url{https://doi.org/10.4230/LIPIcs.ICALP.2021.102}.

\bibitem[PT21b]{Parekh-Thompson-2-local}
Ojas Parekh and Kevin Thompson.
\newblock Beating random assignment for approximating quantum 2-local
  hamiltonian problems.
\newblock In {\em 29th Annual European Symposium on Algorithms (ESA 2021)}.
  Schloss Dagstuhl-Leibniz-Zentrum f{\"u}r Informatik, 2021.
\newblock \url{https://doi.org/10.4230/LIPIcs.ESA.2021.74}.

\bibitem[PT22]{Parekh-Thompson-product}
Ojas Parekh and Kevin Thompson.
\newblock An optimal product-state approximation for 2-local quantum
  hamiltonians with positive terms.
\newblock {\em arXiv preprint arXiv:2206.08342}, 2022.
\newblock \url{https://doi.org/10.48550/arXiv.2206.08342}.

\end{thebibliography}

\pagebreak
\section{Appendix}

\subsection{Proof of Lemma \ref{rotation_lemma_EPR}}\label{rotation_proof_EPR}

\begin{lemma}
\emph{(Restatement of Lemma \ref{rotation_lemma_EPR})}

Consider the state $\prod_{kl} \exp\left(\frac{1}{2} i \theta_{kl} (X_k - Y_k) \otimes (X_l - Y_l)\right) \bigotimes_V |0\rangle$. The energy of this state on the Hamiltonian term $g_{ij}$ is given by
\begin{equation*}
\langle g_{ij} \rangle \geq \frac{1}{2} + \frac{1}{2} \prod_{j\neq k\sim i} \cos{2\theta_{ik}} \cdot \prod_{i\neq l\sim j} \cos{2\theta_{lj}} + \frac{1}{2} \sin{2\theta_{ij}} \left( \prod_{j\neq k\sim i} \cos{2\theta_{ik}} + \prod_{i\neq l\sim j} \cos{2\theta_{lj}} \right)
\end{equation*}
(The notation $j\neq k\sim i$ is shorthand for $k \in \{k \in V : k \neq j , w_{ik} \neq 0\}$.)
\end{lemma}

\begin{figure}[H]
\centering
  \includegraphics[width=3in]{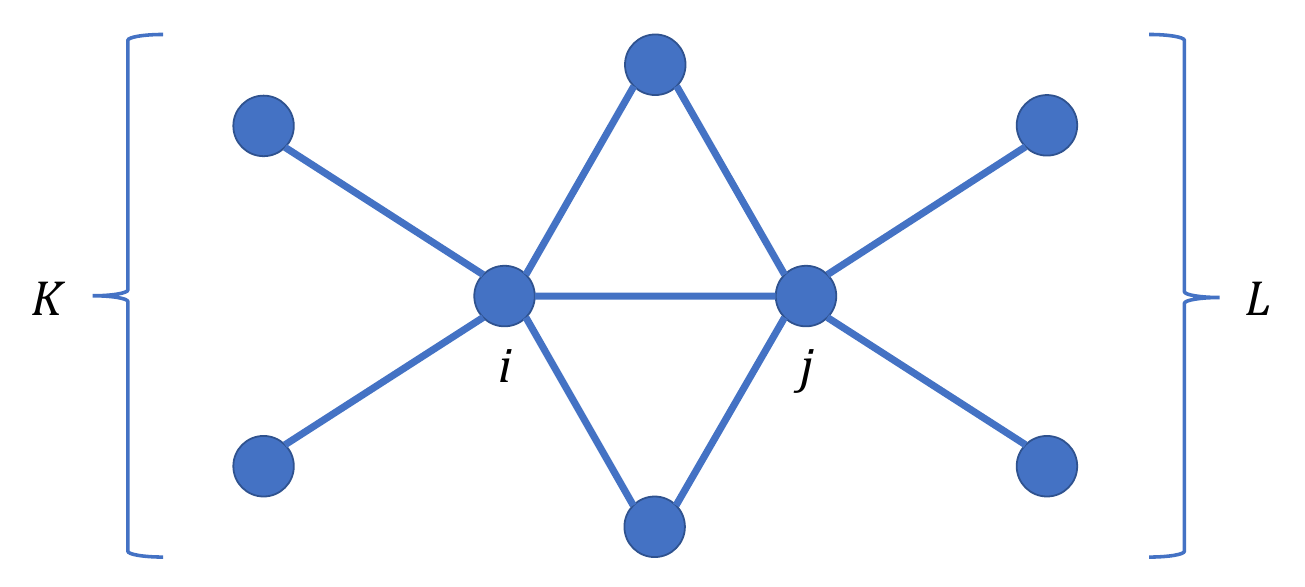}
  \caption{Illustration of $K = \{k\in V : k\sim i, k\neq j\}$ and $L = \{l\in V : l\sim j, l\neq i\}$.}
  \label{fig_rotation}
\
\end{figure}

\begin{proof}
Use the shorthand $P = \frac{X - Y}{\sqrt{2}}$ and $Q = \frac{X + Y}{\sqrt{2}}$. These operators anticommute $PQ = -QP$. Let $K = \{k\in V : k\sim i, k\neq j\}$ and $L = \{l\in V : l\sim j, l\neq i\}$. (See Figure \ref{fig_rotation}.)

It can be checked that
\begin{equation} \label{g_ij}
g_{ij} = \frac{1}{2} \left( \mathbbm{1} + Z_i Z_j + Q_i P_j + P_i Q_j \right)
\end{equation}

We would like to compute the expectation $\langle g_{ij} \rangle$. We will do this by computing the expectation of each term in Equation \ref{g_ij} separately.

\begin{align*}
\langle Q_i P_j \rangle &= \left(\bigotimes_V \langle 0|\right) \left( \prod_{kl} e^{-i \theta_{kl} P_k P_l} \right) Q_i P_j \left( \prod_{kl} e^{i \theta_{kl} P_k P_l} \right) \left(\bigotimes_V |0\rangle\right) \\
&= \left(\bigotimes_V \langle 0|\right) \left( \prod_{k \sim i} e^{-2i \theta_{ik} P_i P_k} \right) Q_i P_j \left(\bigotimes_V |0\rangle\right) \\
&= \left(\bigotimes_V \langle 0|\right) \prod_{k \sim i} \left( \cos{2\theta_{ik}} \mathbbm{1} - i \sin{2\theta_{ik}} P_i P_k \right) Q_i P_j \left(\bigotimes_V |0\rangle\right) \\
&= -i \sin{2\theta_{ij}} \left(\prod_{k\in K} \cos{2\theta_{ik}}\right) \langle 0_i 0_j| P_i P_j Q_i P_j |0_i 0_j\rangle \\
&= \sin{2\theta_{ij}} \prod_{k\in K} \cos{2\theta_{ik}}
\end{align*}
Here we used that $\langle 0|P|0 \rangle = \langle 0|Q|0 \rangle = 0$.

Similarly,
\begin{equation*}
\langle P_i Q_j \rangle = \sin{2\theta_{ij}} \prod_{l\in L} \cos{2\theta_{lj}}
\end{equation*}

Finally,
\begin{align*}
\langle Z_i Z_j \rangle &= \left(\bigotimes_V \langle 0|\right) \left( \prod_{kl} e^{-i \theta_{kl} P_k P_l} \right) Z_i Z_j \left( \prod_{kl} e^{i \theta_{kl} P_k P_l} \right) \left(\bigotimes_V |0\rangle\right) \\
&= \left(\bigotimes_V \langle 0|\right) \left( \prod_{k\in K} e^{-2i \theta_{ik} P_i P_k} \right) \left( \prod_{l\in L} e^{-2i \theta_{lj} P_l P_j} \right) Z_i Z_j \left(\bigotimes_V |0\rangle\right) \\
&= \left(\bigotimes_V \langle 0|\right) \prod_{k\in K} \left( \cos{2\theta_{ik}} \mathbbm{1} - i \sin{2\theta_{ik}} P_i P_k \right) \prod_{l\in L} \left( \cos{2\theta_{lj}} \mathbbm{1} - i \sin{2\theta_{lj}} P_l P_j \right) Z_i Z_j \left(\bigotimes_V |0\rangle\right) \\
&= \left(\bigotimes_V \langle 0|\right) \prod_{k\in K} \left( \cos{2\theta_{ik}} \mathbbm{1} - i \sin{2\theta_{ik}} P_i P_k \right) \prod_{l\in L} \left( \cos{2\theta_{lj}} \mathbbm{1} - i \sin{2\theta_{lj}} P_l P_j \right) \left(\bigotimes_V |0\rangle\right)
\end{align*}

Imagine expanding all brackets in the above expression. Recalling $\langle 0|P|0 \rangle = \langle 0|Q|0 \rangle = 0$, most of the terms will go to zero. Now let $T = \{t\in V : t\sim i, t\sim j\}$; that is, the set of shared neighbours of $i$ and $j$ forming triangles. A pair of triangles $\{i,j,s\}$, $\{i,j,t\}$ will give a term
\begin{equation*}
\sin{2\theta_{is}} \sin{2\theta_{it}} \sin{2\theta_{sj}} \sin{2\theta_{tj}} \prod_{k\in K \setminus \{s,t\}} \cos{2\theta_{ik}} \prod_{l\in L \setminus \{s,t\}} \cos{2\theta_{lj}}
\end{equation*}
More generally, we will pick up a non-zero term precisely when we use the $\sin$ terms for $k$ and $l$ ranging over an \emph{even} number of triangle vertices $s \in S$. This gives
\begin{equation*}
\langle Z_i Z_j \rangle = \sum_{\substack{S \subset T \\ |S| \ \text{even}}} \prod_{s\in S} \sin{2\theta_{is}} \sin{2\theta_{sj}} \prod_{k\in K \setminus S} \cos{2\theta_{ik}} \prod_{l\in L \setminus S} \cos{2\theta_{lj}}
\end{equation*}

Since all these terms are positive, we can throw away all of them except the $S = \emptyset$ term, giving
\begin{equation*}
\langle Z_i Z_j \rangle \geq \prod_{k\in K} \cos{2\theta_{ik}} \prod_{l\in L} \cos{2\theta_{lj}}
\end{equation*}

This completes the proof of Lemma \ref{rotation_lemma_EPR}.
\end{proof}

Lemma \ref{rotation_lemma_EPR} is closely related to Lemma 3 of \cite{AGM}.

\subsection{Proof of Lemma \ref{rotation_lemma_QMC}}\label{rotation_proof_QMC}

\begin{lemma}
\emph{(Restatement of Lemma \ref{rotation_lemma_QMC})}

Given an initial product state $\bigotimes_{k\in V} |v_k\rangle$, choose single-qubit matrices $\{P_k\}_{k\in V}$ and signs $\{\varepsilon_{kl}\}$ as described in Section \ref{rotation_QMC}. Consider the state $\prod_{kl} \exp(i \varepsilon_{kl} \theta_{kl} P_k \otimes P_l) \bigotimes_{m \in V} |v_m\rangle$, where the operators range over all edges $k\neq l$ of a triangle-free graph. The expected energy of this state on the Hamiltonian term $h_{ij}$, averaged over the random choices of $\{P_k\}_{k\in V}$, is given by
\begin{equation*}
\langle h_{ij} \rangle \geq E_{ij} \left( \frac{1}{2} + \frac{1}{2} \prod_{j\neq k\sim i} \cos{2\theta_{ik}} \cdot \prod_{i\neq l\sim j} \cos{2\theta_{lj}} + \frac{1}{\pi} \sin{2\theta_{ij}} \left( \prod_{j\neq k\sim i} \cos{2\theta_{ik}} + \prod_{i\neq l\sim j} \cos{2\theta_{lj}} \right) \right)
\end{equation*}
(The notation $j\neq k\sim i$ is shorthand for $k \in \{k \in V : k \neq j , w_{ik} \neq 0\}$.)
\end{lemma}

\begin{proof}
This proof will use an abuse of notation where $P_k$ will often denote the single qubit matrix $P_k$ acting on qubit $k$ and tensored with identity on all other qubits. The subscript $k\in V$ has two roles: it tells us which matrix is $P_k$, and which qubit it is acting on. Throughout this proof we will also often omit tensor product symbols $\otimes$.

Recall from Section \ref{rotation_QMC} that
\begin{equation*}
2 \langle v_i v_j |\psi^-_{ij}\rangle\langle \psi^-_{ij}|P_i P_j|v_i v_j\rangle = e^{-i \gamma_{ij}} E_{ij}
\end{equation*}
where
\begin{equation*}
\gamma_{ij} = \pi - \arg{\langle v_j|P_i|v_i\rangle\langle v_i|P_j|v_j\rangle}
\end{equation*}

From Lemma \ref{equator_lem} $\gamma_{ij} \in [-\pi,\pi)$ is uniformly random. Now if we choose the sign $\varepsilon_{ij} \in \{+1,-1\}$ according to Equation \ref{sign_choice}, that is
\begin{equation*}
\gamma_{ij} =
	\begin{cases}
	[0,\pi) & \rightarrow \ \varepsilon_{ij} = +1 \\
	[-\pi,0) & \rightarrow \ \varepsilon_{ij} = -1
	\end{cases}
\end{equation*}
then we will have
\begin{equation} \label{phased_overlap}
2 \langle v_i v_j |\psi^-_{ij}\rangle\langle \psi^-_{ij}| \varepsilon_{ij} P_i P_j|v_i v_j\rangle = e^{-i \gamma'_{ij}} E_{ij}
\end{equation}
with $\gamma'_{ij} \in [0,\pi)$ uniformly random.

\bigskip
Now let $K = \{k\in V : k\sim i, k\neq j\}$ and $L = \{l\in V : l\sim j, l\neq i\}$. (See Figure \ref{fig_rotation}.) We can write
\begin{equation} \label{h_ij}
\langle h_{ij} \rangle = \left(\bigotimes_{k\in V} \langle v_k|\right) U^\dag h_{ij} U \left(\bigotimes_{k\in V} |v_k\rangle\right)
\end{equation}
where
\begin{equation*}
U =  e^{i \varepsilon_{ij} \theta_{ij} P_i P_j} \left( \prod_{k \in K} e^{i \varepsilon_{ik} \theta_{ik} P_i P_k} \right) \left( \prod_{l \in L} e^{i \varepsilon_{lj} \theta_{lj} P_l P_j} \right)
\end{equation*}
since the other gates in the product $\prod_{kl} \exp(i \varepsilon_{ij} \theta_{kl} P_k P_l)$ commute with $h_{ij}$.

We can rewrite $U$ as
\begin{align}
U &= e^{i \varepsilon_{ij} \theta_{ij} P_i P_j} \prod_{k\in K, l\in L} (\cos{\theta_{ik}} \mathbbm{1} + i \varepsilon_{ij} \sin{\theta_{ik}} P_i P_k) (\cos{\theta_{lj}} \mathbbm{1} + i \varepsilon_{ij} \sin{\theta_{lj}} P_l P_j) \nonumber \\
&= \sum_{S\subset K, T\subset L} \left( \prod_{k\in K\setminus S} \cos{\theta_{ik}} \right) \left( \prod_{l\in L\setminus T} \cos{\theta_{lj}} \right) \left( \prod_{k\in S} i \varepsilon_{ij} \sin{\theta_{ik}} \right) \left( \prod_{l\in T} i \varepsilon_{ij} \sin{\theta_{lj}} \right) \nonumber\\
&\qquad\qquad\qquad\qquad\qquad\qquad\qquad\qquad\qquad \left( \prod_{k\in S} P_k \right) \left( \prod_{l\in T} P_l \right) e^{i \varepsilon_{ij} \theta_{ij} P_i P_j} P_i^{|S|} P_j^{|T|} \label{U_expansion_QMC}
\end{align}

Let's abbreviate the expression
\begin{equation*}
\alpha(S,T) = \left( \prod_{k\in K\setminus S} \cos{\theta_{ik}} \right) \left( \prod_{l\in L\setminus T} \cos{\theta_{lj}} \right) \left( \prod_{k\in S} i \varepsilon_{ij} \sin{\theta_{ik}} \right) \left( \prod_{l\in T} i \varepsilon_{ij} \sin{\theta_{lj}} \right)
\end{equation*}

The aim is to calculate $\langle h_{ij} \rangle$. We do this by expanding $U$ and $U^\dag$ in Equation \ref{h_ij} using Equation \ref{U_expansion_QMC}. This will give us a sum over $\{S \subset K, T \subset L, S' \subset K, T' \subset L\}$, where $S'$ and $T'$ come from the expansion of $U^\dag$. Recall $h_{ij} = 2 |\psi^-_{ij}\rangle\langle\psi^-_{ij}|$.

\begin{align*}
\langle h_{ij} \rangle &= \left(\bigotimes_{k\in V} \langle v_k|\right) U^\dag h_{ij} U \left(\bigotimes_{k\in V} |v_k\rangle\right) \\
&= \sum_{S,S'\subset K, T,T'\subset L} \alpha(S,T) \overline{\alpha(S',T')} \left(\bigotimes_{k\in V\setminus \{i,j\}} \langle v_k|\right) \prod_{k'\in S'} P_{k'} \prod_{l'\in T'} P_{l'} \prod_{k\in S} P_k \prod_{l\in T} P_l \left(\bigotimes_{k\in V\setminus \{i,j\}} |v_k\rangle\right) \\
&\qquad\qquad\qquad\qquad\qquad \cdot 2 \langle v_i v_j| P_i^{|S'|} P_j^{|T'|} e^{-i \varepsilon_{ij} \theta_{ij} P_i P_j} |\psi^-_{ij}\rangle\langle\psi^-_{ij}| e^{i \varepsilon_{ij} \theta_{ij} P_i P_j} P_i^{|S|} P_j^{|T|} |v_i v_j\rangle
\end{align*}

Consider the expression
\begin{equation} \label{expression_QMC}
\left(\bigotimes_{k\in V\setminus \{i,j\}} \langle v_k|\right) \prod_{k'\in S'} P_{k'} \prod_{l'\in T'} P_{l'} \prod_{k\in S} P_k \prod_{l\in T} P_l \left(\bigotimes_{k\in V\setminus \{i,j\}} |v_k\rangle\right)
\end{equation}
Since by assumption $i$ and $j$ share no neighbours, this expression vanishes unless $S=S'$ and $T=T'$. Thus in the expansion of \ref{h_ij}, we can ignore the `off-diagonal' terms $(S,T) \neq (S',T')$. This is the only place in the paper where we use the triangle-free assumption.

\begin{align*}
\langle h_{ij} \rangle &= \left(\bigotimes_{k\in V} \langle v_k|\right) U^\dag h_{ij} U \left(\bigotimes_{k\in V} |v_k\rangle\right) \\
&= \sum_{S\subset K, T\subset L} |\alpha(S,T)|^2 \cdot 2 \left|\langle\psi^-_{ij}| e^{i \varepsilon_{ij} \theta_{ij} P_i P_j} P_i^{|S|} P_j^{|T|} |v_i v_j\rangle \right|^2 \\
&= \sum_{S\subset K, T\subset L} |\alpha(S,T)|^2 \cdot 2 \left|\langle\psi^-_{ij}| \left( \cos{\theta_{ij}} P_i^{|S|} P_j^{|T|} + i \varepsilon_{ij} \sin{\theta_{ij}} P_i^{|S|+1} P_j^{|T|+1} \right) |v_i v_j\rangle \right|^2 \\
&= \sum_{S\subset K, T\subset L} |\alpha(S,T)|^2 \cdot
	\begin{cases}
	2 \left|\langle\psi^-_{ij}| \left(\cos{\theta_{ij}} |v_i v_j\rangle + i \varepsilon_{ij} \sin{\theta_{ij}} P_i P_j |v_i v_j\rangle \right)\right|^2 & |S|, |T| \ \text{even} \\
	2 \left|\langle\psi^-_{ij}| \left(\cos{\theta_{ij}} P_i P_j |v_i v_j\rangle + i \varepsilon_{ij} \sin{\theta_{ij}} |v_i v_j\rangle \right)\right|^2 & |S|, |T| \ \text{odd} \\
	\geq 0 & \text{otherwise}
	\end{cases} \\
&= \sum_{S\subset K, T\subset L} |\alpha(S,T)|^2 \cdot
	\begin{cases}
	2 \cos^2{\theta_{ij}} \langle v_i v_j|\psi^-_{ij}\rangle\langle \psi^-_{ij}|v_i v_j\rangle \\
	\quad + 2 \sin^2{\theta_{ij}} \langle v_i v_j|P_i P_j|\psi^-_{ij}\rangle\langle \psi^-_{ij}|P_i P_j|v_i v_j\rangle \\
    \quad + 2i \sin{\theta_{ij}} \cos{\theta_{ij}} \langle v_i v_j|\psi^-_{ij}\rangle\langle \psi^-_{ij}|\varepsilon_{ij} P_i P_j|v_i v_j\rangle \\
    \quad - 2i \sin{\theta_{ij}} \cos{\theta_{ij}} \langle v_i v_j|\varepsilon_{ij} P_i P_j|\psi^-_{ij}\rangle\langle \psi^-_{ij}|v_i v_j\rangle & |S|, |T| \ \text{even} \\
	2 \cos^2{\theta_{ij}} \langle v_i v_j|P_i P_j|\psi^-_{ij}\rangle\langle \psi^-_{ij}|P_i P_j|v_i v_j\rangle \\
	\quad + 2 \sin^2{\theta_{ij}} \langle v_i v_j|\psi^-_{ij}\rangle\langle \psi^-_{ij}|v_i v_j\rangle \\
    \quad - 2i \sin{\theta_{ij}} \cos{\theta_{ij}} \langle v_i v_j|\psi^-_{ij}\rangle\langle \psi^-_{ij}|\varepsilon_{ij} P_i P_j|v_i v_j\rangle \\
    \quad + 2i \sin{\theta_{ij}} \cos{\theta_{ij}} \langle v_i v_j|\varepsilon_{ij} P_i P_j|\psi^-_{ij}\rangle\langle \psi^-_{ij}|v_i v_j\rangle & |S|, |T| \ \text{odd} \\
	\geq 0 & \text{otherwise}
	\end{cases} \\
&= \sum_{S\subset K, T\subset L} |\alpha(S,T)|^2 \cdot
	\begin{cases}
	E_{ij} + E_{ij} \sin{\theta_{ij}} \cos{\theta_{ij}} (i e^{-i \gamma'_{ij}} - i e^{i \gamma'_{ij}}) & |S|, |T| \ \text{even} \\
	E_{ij} + E_{ij} \sin{\theta_{ij}} \cos{\theta_{ij}} (- i e^{-i \gamma'_{ij}} + i e^{i \gamma'_{ij}}) & |S|, |T| \ \text{odd} \\
	\geq 0 & \text{otherwise}
	\end{cases} \\
&= \sum_{S\subset K, T\subset L} |\alpha(S,T)|^2 \cdot
	\begin{cases}
	E_{ij} (1 + \sin{\gamma'_{ij}} \sin{2\theta_{ij}}) & |S|, |T| \ \text{even} \\
	E_{ij} (1 - \sin{\gamma'_{ij}} \sin{2\theta_{ij}}) & |S|, |T| \ \text{odd} \\
	\geq 0 & \text{otherwise}
	\end{cases}
\end{align*}

To deal with the expectation, we average over $\gamma'_{ij} \in [0,\pi)$.
\begin{align*}
\mathbb{E} \langle h_{ij} \rangle &= \sum_{S\subset K, T\subset L} |\alpha(S,T)|^2 \cdot
	\begin{cases}
	E_{ij} (1 + \mathbb{E}_{\gamma'_{ij}} \sin{\gamma'_{ij}} \sin{2\theta_{ij}}) & |S|, |T| \ \text{even} \\
	E_{ij} (1 - \mathbb{E}_{\gamma'_{ij}} \sin{\gamma'_{ij}} \sin{2\theta_{ij}}) & |S|, |T| \ \text{odd} \\
	\geq 0 & \text{otherwise}
	\end{cases} \\
&= \sum_{S\subset K, T\subset L} |\alpha(S,T)|^2 \cdot
	\begin{cases}
	E_{ij} (1 + \frac{2}{\pi} \sin{2\theta_{ij}}) & |S|, |T| \ \text{even} \\
	E_{ij} (1 - \frac{2}{\pi} \sin{2\theta_{ij}}) & |S|, |T| \ \text{odd} \\
	\geq 0 & \text{otherwise}
	\end{cases} \\
&\geq E_{ij} \left(\left( \sum_{|S|+|T| \ \text{even}} |\alpha(S,T)|^2 \right) + \frac{2}{\pi} \sin{2\theta_{ij}} \left( \sum_{|S|+|T| \ \text{even}} (-1)^{|S|} \cdot |\alpha(S,T)|^2 \right)\right)
\end{align*}

Now consider expanding the brackets in the following expression.
\begin{equation*}
\prod_{k\in K} (\cos^2{\theta_{ik}} + \sin^2{\theta_{ik}}) \prod_{l\in L} (\cos^2{\theta_{lj}} + \sin^2{\theta_{lj}}) + \prod_{k\in K} (\cos^2{\theta_{ik}} - \sin^2{\theta_{ik}}) \prod_{l\in L} (\cos^2{\theta_{lj}} - \sin^2{\theta_{lj}})
\end{equation*}
We get precisely
\begin{align*}
\prod_{k\in K} (\cos^2{\theta_{ik}} + \sin^2{\theta_{ik}}) \prod_{l\in L} (\cos^2{\theta_{lj}} + \sin^2{\theta_{lj}}) &+ \prod_{k\in K} (\cos^2{\theta_{ik}} - \sin^2{\theta_{ik}}) \prod_{l\in L} (\cos^2{\theta_{lj}} - \sin^2{\theta_{lj}}) \\
&= 2 \sum_{|S|+|T| \ \text{even}} |\alpha(S,T)|^2
\end{align*}
and similarly
\begin{align*}
\prod_{k\in K} (\cos^2{\theta_{ik}} + \sin^2{\theta_{ik}}) \prod_{l\in L} (\cos^2{\theta_{lj}} - \sin^2{\theta_{lj}}) &+ \prod_{k\in K} (\cos^2{\theta_{ik}} - \sin^2{\theta_{ik}}) \prod_{l\in L} (\cos^2{\theta_{lj}} + \sin^2{\theta_{lj}}) \\
&= 2 \sum_{|S|+|T| \ \text{even}} (-1)^{|S|} \cdot |\alpha(S,T)|^2
\end{align*}

Thus
\begin{align*}
\mathbb{E} \langle h_{ij} \rangle &\geq E_{ij} \Bigg( \frac{1}{2} \Bigg( \prod_{k\in K} (\cos^2{\theta_{ik}} + \sin^2{\theta_{ik}}) \prod_{l\in L} (\cos^2{\theta_{lj}} + \sin^2{\theta_{lj}}) \nonumber\\
&\qquad\qquad + \prod_{k\in K} (\cos^2{\theta_{ik}} - \sin^2{\theta_{ik}}) \prod_{l\in L} (\cos^2{\theta_{lj}} - \sin^2{\theta_{lj}}) \Bigg) \nonumber\\
&\qquad\quad + \frac{1}{\pi} \sin{2\theta_{ij}} \Bigg( \prod_{k\in K} (\cos^2{\theta_{ik}} + \sin^2{\theta_{ik}}) \prod_{l\in L} (\cos^2{\theta_{lj}} - \sin^2{\theta_{lj}}) \nonumber\\
&\qquad\qquad\qquad\qquad + \prod_{k\in K} (\cos^2{\theta_{ik}} - \sin^2{\theta_{ik}}) \prod_{l\in L} (\cos^2{\theta_{lj}} + \sin^2{\theta_{lj}}) \Bigg) \Bigg) \\
&= E_{ij} \left(\frac{1}{2} \left( 1 + \prod_{k\in K} \cos{2\theta_{ik}} \prod_{l\in L} \cos{2\theta_{lj}} \right) + \frac{1}{\pi} \sin{2\theta_{ij}} \left( \prod_{k\in K} \cos{2\theta_{ik}} + \prod_{l\in L} \cos{2\theta_{lj}} \right)\right)
\end{align*}
\end{proof}

\subsection{Proof of Corollary \ref{starbound_corollary}}\label{starbound_proof}

\begin{lemma} \label{starbound_corollary_lemma}
Let $\{x_k : k \in S\}$ be positive real numbers indexed by $S$, and $\eta \in [0,1]$. If $\sum_{k\in S} x_k \leq \eta$, then for any $\beta \in [0,1]$ we have
\begin{equation*}
\prod_{k\in S} \sqrt{1 - \beta^2 x_k^2} \geq \sqrt{1 - \beta^2 \eta^2}
\end{equation*}
\end{lemma}
\begin{proof}
The claim is that the minimum of the expression is achieved by setting $x_{k^*} = \eta$ for some $k^* \in S$ and $x_{ik}=0 \ \forall \ k \in S, k \neq k^*$. That is, put all the mass on one edge. To see this, suppose $x_{k'}, x_{k''}>0$. Fix $\{x_k : k \in S \setminus \{k',k''\}\}$. The expression depends on $\{x_{k'},x_{k''}\}$ solely through $\sqrt{1-\beta^2x_{k'}^2}\cdot\sqrt{1-\beta^2x_{k''}^2}$, subject to $x_{k'} + x_{k''} \leq \eta - \sum_{k \in S \setminus \{k',k''\}} x_{k}$. The expression can be decreased by setting $x_{k'} \leftarrow x_{k'} + x_{k''}$, $x_{k''} \leftarrow 0$.
\end{proof}

We now prove Corollary \ref{starbound_corollary}. For Part (1), set $S = \{k : k \sim i, x_{ik} \geq 0\}$ and $\eta = 1$. For Part (2), set $S = \{k : k \sim i, k \neq j, x_{ik} \geq 0\}$ and $\eta = 1 - x_{ij}$. In both cases, the conclusions follow from Lemma \ref{starbound_lemma} and Lemma \ref{starbound_corollary_lemma}. The arguments for $y_{ij}$ are identical.

\end{document}